\documentclass[11pt]{article}
\pdfoutput=1
\usepackage{amsmath}
\usepackage{amssymb}
\usepackage{amsthm}
\usepackage{color}
\usepackage{graphicx}
\usepackage{epstopdf}
\usepackage{authblk}
\usepackage{natbib}
\usepackage{dsfont}
\usepackage{booktabs}
\usepackage[margin=1in]{geometry}
\usepackage[font=footnotesize,labelfont=bf,skip=1pt]{caption}
\usepackage[colorlinks=true, linkcolor=blue, citecolor=blue, filecolor=magenta, urlcolor=blue]{hyperref}

\newcommand{\norm}[2][2]{\|{#2}\|_{#1}}
\newcommand{\Ex}[2][]{\ensuremath{\mathbb{E}_{#1}\left[#2\right]}}
\newcommand{\trace}[1]{\ensuremath{\operatorname{tr}\left(#1\right)}}
\newcommand{\RL}{\mathds{R}}
\newtheorem{theorem}{Theorem}
\newtheorem{definition}[theorem]{Definition}
\newtheorem{lemma}[theorem]{Lemma}

\newtheorem{corollary}[theorem]{Corollary}
\newtheorem{observation}{Observation}

\begin{document}

\setcitestyle{aysep={~},yysep={;}}

\title{Random projections for Bayesian regression
}


\author
       {Leo N. Geppert$^1$, Katja Ickstadt$^1$,\\ \vspace{-0.3cm}
       Alexander Munteanu$^{2}$, Jens Quedenfeld$^2$, and Christian Sohler$^2$
       \\ \vspace{0.3cm}
       $^1$ Department of Statistics\\
       \texttt{\{geppert,\,ickstadt\}@statistik.uni-dortmund.de}       \\ 

       $^2$ Department of Computer Science\\              
              \texttt{\{alexander.munteanu,\,jens.quedenfeld,\,christian.sohler\}@tu-dortmund.de} \\
              \vspace{0.4cm}
              
        Technische Universit\"{a}t Dortmund\\ 44221 Dortmund, Germany\\
       }


%
\date{November 30, 2015}

\maketitle


\begin{abstract}
This article deals with random projections applied as a data reduction technique for Bayesian regression analysis. We show sufficient conditions under which the entire $d$-dimensional distribution is approximately preserved under random projections by reducing the number of data points from $n$ to $k\in O(\operatorname{poly}(d/\varepsilon))$ in the case $n\gg d$. Under mild assumptions, we prove that evaluating a Gaussian likelihood function based on the projected data instead of the original data yields a $(1+O(\varepsilon))$-approximation in terms of the $\ell_2$ Wasser\-stein distance. Our main result shows that the posterior distribution of Bayesian linear regression is approximated up to a small error depending on only an $\varepsilon$-fraction of its defining parameters. This holds when using arbitrary Gaussian priors or the degenerate case of uniform distributions over $\mathds{R}^d$ for $\beta$. Our empirical evaluations involve different simulated settings of Bayesian linear regression. Our experiments underline that the proposed method is able to recover the regression model up to small error while considerably reducing the total running time.
\end{abstract}



\section{Introduction}
\label{sec:intro}
In this paper we consider linear regression. Using a linear map $\Pi \in\RL^{k\times n}$ whose choice is still to be defined, we transform the original data set $[X,Y]\in \RL^{n\times (d+1)}$ into a sketch, i.e., a substitute data set, $[\Pi X,\Pi Y]\in\RL^{k\times (d+1)}$ that is considerably smaller. Therefore, the likelihood function can be evaluated faster than on the original data. Moreover, we will show that the likelihood is very similar to the original one. In the context of Bayesian regression we have the likelihood $\mathcal L(\beta | X, Y)$ and additional prior information $p_{\rm pre}(\beta)$ given in form of a prior distribution over the parameters $\beta\in\RL^d$ which we would like to estimate. Our main result is to show that the resulting posterior distribution
\begin{eqnarray}
p_{\rm post}(\beta|X, Y)&\propto f(Y | \beta, X)\cdot p_{\rm pre}(\beta) \label{eq:Bayestheorem}
\end{eqnarray}
will also be well approximated within a small error. Please note that Bayes' Theorem (\ref{eq:Bayestheorem}) contains the probability density function $f(Y | \beta, X)$. From now on, we will concentrate on its interpretation as likelihood function $\mathcal L(\beta | X, Y)$ as a function of the unknown parameter vector $\beta$ as given in (\ref{eq:Bayestheorem2})
\begin{eqnarray}
p_{\rm post}(\beta|X, Y)&\propto \mathcal L(\beta | X,Y)\cdot p_{\rm pre}(\beta).
\label{eq:Bayestheorem2}
\end{eqnarray}

The main idea of our approach is given in the following scheme:
$$
\begin{array}{rcl}
[X,Y] & \xrightarrow{\quad \Pi\quad} & \,\,[\Pi X,\Pi Y]\\
\downarrow \quad & &\quad\quad \downarrow\\
p_{\rm post}(\beta | X, Y) \hspace*{-.2cm} & \approx_\varepsilon & \hspace*{-.3cm} p_{\rm post}(\beta | \Pi X,\Pi Y).
\end{array}
$$

More specifically, we can reduce the number of observations from the number of input points $n$ to a target dimension $k\in O(\operatorname{poly}(d/\varepsilon))$, which, in particular, is independent of $n$. Thus, the running time of all subsequent calculations does not further depend on $n$. For instance, a Markov Chain Monte Carlo (MCMC) sampling algorithm may be used to obtain samples from the unknown distribution. Using the reduced data set will speed up the computations considerably. The samples remain sufficiently accurate to resemble the original distribution and also to make statistical predictions that are nearly indistinguishable from the predictions that would have been made based on the original sample. Note, that mathematically it is possible to achieve a similar reduction by setting $\Pi=X^T$ without incurring any error. From the resulting matrix $[X^TX, X^TY]$ it would be possible to compute or evaluate the exact likelihood respectively posterior in some analytically tractable cases, which is the standard textbook approach for classical linear regression and Bayesian linear regression with Gaussian prior and independent normal error, \citep[cf.][]{Bishop:2006, Hastie2009}. However, this approach is not numerically stable leading to ill-conditioned covariance matrices, which is known from the literature \citep[cf.][]{Golub1965, LawsonH1995, GolubVanLoan2013} and evident from our experiments. For that reason, the exact calculation is not an option in general. Instead, approximation algorithms are an alternative, where the error can be controlled.

Using no reduction technique is also not an option, since then even likelihood evaluations depend at least linearly on $n$. This does not pose a problem for small data sets. For larger $n$ this is also still possible, but employing reduction techniques can already be beneficial to reduce the running time. For data sets that do not fit into the working memory, intelligent solutions are needed to avoid frequent swapping to slower secondary memory. However, for really massive data or infinite data streams, already the \emph{much simpler} problem of computing the exact ordinary least squares estimator in one pass over the data requires at least $\Omega(n)$ space \citep{CW09}. This makes the task impossible on finite memory machines when $n$ grows large enough.

There are different computational models to deal with massive data sets in streaming and distributed environments. We focus on the streaming model, formally introduced in \cite{Muthukrishnan05}. A data stream algorithm is given an input stream of items, like numerical values, points in $\mathds{R}^d$ or edges of a graph at a high rate. As the items arrive one by one it maintains some sort of summary of the data that is observed so far. This can be a subsample or a linear sketch as described above. The linearity allows for flexible dynamic updates of the sketch as we will discuss later. At any time, the memory used by the algorithm is restricted to be sublinear, usually at most polylogarithmic in the number of items. For geometrical problems the dependence on the dimension $d$ is often restricted to be at most a small polynomial. Also, the algorithm is allowed to make only one single pass over the data. With these restrictions in mind, it is clear that the sketching matrix $\Pi\in\mathds{R}^{k \times n}$ cannot be explicitly stored in memory. In fact it has to fulfill the following criteria to allow for a streaming algorithm:

\begin{enumerate}
	\item $\Pi[X,Y]$ approximates $[X,Y]$ well in the above sense.
	\item $\Pi$ can be stored succinctly.
	\item We can efficiently generate the entries of $\Pi$.
\end{enumerate}

We will see that the structured randomized constructions of $\Pi$ can be provably achieved using random bits of only limited independence. This means that the entries of $\Pi$ need not be fully independent. However, if we choose a small number of entries, they behave as if they were independent. In particular the entries can be stored implicitly, meeting the independence requirements by using hash functions. These can be evaluated very efficiently and make the memory dependency on $n$ only logarithmic.

Although the techniques presented in the following are employed to make the computations possible in a streaming setting, the results are of interest also in the non-streaming setting whenever large data sets can be reduced to meet time and memory constraints.


\section{Background and Related Work}
\label{sec:background}
Dimensionality reduction techniques like principal component analysis (PCA) \citep{Jolliffe:2002} and random projections have been widely used in Statistics and Computer Science. However, their focus is usually on reducing the number of variables. Our method aims to reduce the number of observations while keeping the algebraic structure of the data. This leads to a speed-up in the subsequent (frequentist or Bayesian) regression analysis, because the running times of commonly used algorithms heavily depend on $n$. Basic techniques based on PCA include principle component regression and partial least squares \citep{Hastie2009}. More recent results using PCA stem from the theory of core-sets for the $k$-means clustering problem and address the problem of computing a small set of data that approximates the original point set with respect to the given objective up to little, say $(1\pm\varepsilon)$, error \citep{FeldmanSS13}. The concept of core-sets is related to the early works of \citet{MadiganRDNPR02} and \citet{DuMouchelVJCP99} on \emph{data-squashing} that seeks for data compression based on the likelihood of the observations. One of the more recent contributions is \citet{quiroz2014}. They suggest inclusion probabilities proportional to the observations' contribution to the likelihood, which is approximated by a Gaussian Process or a thin-plate spline approximation. Data-squashing can lead to a considerable reduction in the necessary number of observations, however there is a lack of approximation guarantees. These references show that in the advent of massive data sets, besides the effort in reducing dimensionality, there is also need to reduce the number of observations without incurring loss of too much statistical information.

Random projections have been studied in the context of low-rank approximation\linebreak \citep{CohenEMMP15}, least squares regression \citep{Sarlos2006,CW09}, Gaussian process regression \citep{BanerjeeDT12}, clustering problems \citep{BoutsidisZD10, KerberR14, CohenEMMP15}, classification tasks \citep{PaulBMD14} and compressed sensing \citep{CandesRT06,Donoho06}. Random projections are used similarly to our work, to approximate a collection of subspaces, consisting only of sparse vectors \citep{Baraniuk07}. Also, Bayesian inference has been proposed for efficient computation in compressed sensing  \citep{JiC07}.

Recently there has been a series of works dealing with the statistical aspects of randomized linear algebra al\-go\-rithms. In \cite{RaskuttiM15} and \cite{MaMY14}, the statistical properties of subsampling approaches based on the \emph{statistical leverage scores} of the data are investigated in detail. Deviating from the worst case algorithmic perspective, it was shown in \cite{MaMY14} that on average the leverage scores behave quite uniformly if the data is generated following a standard linear regression model. In \cite{YangMM15}, several sketching and subsampling methods are used for fast preconditioning before solving the ordinary least squares (OLS) estimators on the subsampled data using state of the art OLS solvers. Moreover, they give parallel and distributed algorithms and extensive empirical evaluations on large scale data for this task. Our work, while aware of providing worst case guarantees, continues the discussion of statistical properties to the Bayesian setting.

Bayesian regression analysis for large scale data sets has been considered before. \citet{Dunson2013} proposed reducing the number of variables via random projections as a preprocessing step in the large $d$, small $n$ scenario. They show that under several assumptions the approximation converges to the desired posterior distribution, which is not possible in general, since it was shown in \cite{BoutsidisM14} that such dimensionality reduction oblivious to the target variable causes additive error in the worst case. 

For the large $n$, small $d$ case, \citet{balakrishnan2006} proposed a one-pass algorithm that reads the data block-wise and performs a certain number of MCMC steps. When the next block is read, the algorithm keeps or replaces some of the data points based on weights that keep track of the importance of the data. The selection rule is justified empirically but lacks theoretical guarantees. Theoretical support is only given in the univariate case based on central limit theorems for Sequential Monte Carlo methods.

When allowing more passes over the data, TSQR \citep{DemmelGHL12} is a QR decomposition which works especially well in the large $n$, small $d$ case and can easily be parallelized in the MapReduce setting, as studied by \cite{ConstantineG11} and \cite{BensonGD13}. TSQR provides a numerically stable decomposition, which can be used as a preprocessing step prior to MCMC inference which can be conducted with high accuracy. This method, however, depending on the computational setting, has a number of limitations. It only works when the data is given row-by-row and the expensive decomposition has to be carried out a linear number of times, resulting in a total lower bound on the running time of $\Omega(n d^2)$ \citep[cf.][]{DemmelGHL12}. The method is restricted to $\ell_2$ regression. Our method of random projections is capable of going beyond these limitations. In particular, it can be extended to $\ell_p$ regression and more generally to robust \emph{M}-estimators \citep{CW15}. This flexibility comes at the price of a loss in accuracy, however, this loss is controllable by bounding parameters and does not lead to invalid inference.

Another approach by \cite{Bardenet14} tries to subsample the data to approximate the acceptance rule in each iteration of an MCMC sampler. The decision is shown to be similar to the original with high probability in each step. The number of samples is highly dependent on the variance of the logarithm of likelihood ratios. The method may be useful for interesting and intractable cases when the variance can be bounded.

While frequentist linear regression can be solved straightforwardly by computing the projection of the target variable to the subspace spanned by the data, Bayesian regression is typically computationally more demanding. In some cases, it is possible to calculate the posterior distribution analytically, but in general this is not true. For that reason, an approximation of the posterior distribution is needed. MCMC methods are one possibility and standard in Bayesian analysis. They are reliable, but can take considerable time before they converge and sample from the desired posterior distribution. Moreover, the running time grows with the number of observations in the data set.

The main bottleneck of a lot of Bayesian analysis methods including MCMC is the repeated evaluation of the likelihood. The running time of each evaluation grows linearly with the number of observations in the data set. There are several approaches trying to reduce the computational effort for Bayesian regression analysis by employing different algorithms that may perform more efficient in certain settings. Approximate Bayesian Computing (ABC) and Integrated Nested Laplace Approximations (INLA) both fall into this category. The main idea behind ABC is to avoid the exact evaluations by approximating the likelihood function using simulations \citep{abc:csillery:2010}. INLA \citep{Rue:INLA, Martinsetal:2013} on the other hand is an approximation of the posterior distribution that is applicable to models that fall into the class of so-called latent Gaussian models. Both methods can lead to a considerable speed-up compared to standard MCMC algorithms.

Note however, that the speed-up is achieved by changing the algorithm which is used to conduct the analysis. This is different in our approach, which reduces the number of observations in the data set while approximately retaining its statistical properties. The running times of many algorithms including MCMC algorithms highly depend on the number of observations, which means that our proposed method also results in a speed-up of the analysis. In this article, we focus on MCMC methods for the analysis, but in principle, as our method provably approximates the posterior, all algorithms that assess the posterior can be employed. We did not use ABC since it is only suitable for summary statistics of very low dimension ($d <10$) \citep{BeaumontABC:2002, abc:csillery:2010}. However, we have tried INLA on a small scale, achieving comparable results as with MCMC, making the running time of the analysis independent of $n$. Likewise one could consider calculating the exact formulae for analytically tractable cases of the posterior. However, we concentrate on MCMC methods because of their general applicability and reliability.

New directions in Bayesian data analysis in the context of Big Data are surveyed in \cite{welling:2014}. Our work directly suits the criteria that is proposed in that reference for the large $n$ case in streaming as well as distributed computation environments.


\section{Preliminaries}
\label{sec:prelim}

\subsection{General notation}

For the sake of a brief presentation, we introduce some notation. We denote by $[n]=\{1,\ldots,n\}$ the set of all positive integers up to $n$. For a probability measure $\lambda$, let $\Ex[\lambda]{X}$ be the expected value of $X$ with respect to $\lambda$. We skip the subscript in $\Ex{X}$ if the probability measure is clear from the context. For a matrix $M\in \RL^{n\times d}$ we let $M=U\Sigma V^T$ denote its singular value decomposition (SVD), where $U\in \RL^{n\times d}$ and $V\in\RL^{d\times d}$ are unitary matrices spanning the column\-space and rowspace of $M$, respectively. $\Sigma\in \RL^{d\times d}$ is a diagonal matrix, whose elements $\sigma_1\geq\ldots\geq\sigma_d$ are the singular values of $M$. We denote by $\sigma_{\rm max}=\sigma_1$ the largest and by $\sigma_{\rm min}=\sigma_d$ the smallest singular value of $M$ and write $\sigma_i(M)$ to make clear the $\sigma_i$ belong to $M$. The trace of $M^TM$ equals the sum of the squared singular values, i.e., $\trace{M^TM}=\sum_{i=1}^d \sigma_i^2(M)$. We assume w.l.o.g. all matrices to have full rank and stress that all our proofs carry out similarly to our presentation if the matrices are of lower rank. One might even use knowledge about lower rank to reduce the space and time complexities to bounds that only depend on the rank rather than on the number of variables.

\subsection{Bayesian regression}
\label{sec:prelim:bayesreg}

A linear regression model is given in the following equation:
\begin{equation*}
Y = X\beta + \xi.
\end{equation*}

$Y \in \RL^n$ is a random variable containing the values of the response, where $n$ is the number of observations in the data set. $X \in \RL^{n\times d}$ is a matrix containing the values of the $d$ independent variables. We denote by $\xi \sim N(0, \varsigma^2 I_n)$ an $n$-dimensional random vector that models the unobservable error term. The dependent variable $Y$ then follows a Gaussian distribution, \mbox{$Y \sim N(X\beta, \varsigma^2 I_n)$}. The corresponding probability density function is
\begin{eqnarray*}
f(y | X\beta, \Sigma) &=& (2\pi)^{-\frac{n}{2}} \left\vert \Sigma \right\vert ^{-\frac{1}{2}} \exp{\left(-\frac{1}{2\varsigma^2}\norm{X\beta-y}^2\right)},
\end{eqnarray*}
where $\Sigma = \varsigma^2 I_n$.

In a Bayesian setting, $\beta \in \RL^d$ is the unknown parameter vector which is assumed to follow an unknown distribution $p(\beta|X,Y$) called the posterior distribution. Prior knowledge about $\beta$ can be modeled using the prior distribution $p(\beta)$. The posterior distribution is a compromise between the prior distribution and the observed data. 

In general, the posterior distribution cannot be calculated analytically. In this paper, we determine the posterior distribution employing Markov Chain Monte Carlo methods, even though the posterior is explicitly known. Regardless of the computational problems related to these explicit formulae (cf. Section \ref{sec:intro}), we focus on MCMC, because this work forms the basis for further research on more complex models where analytical solutions are not obtainable. Possible extensions are hierarchical models and mixtures of normal distributions (cf. Section \ref{sec:conclusion}). Furthermore, we follow this strategy to rule out possible interaction effects between sketching and MCMC that might occur even in these basic cases. However, our empirical evaluation indicates that there are none.

\subsection{Norms and metrics}

Before going into details about random projections and subspace embeddings, let us first define the matrix and vector norms used in this paper as well as the metric that we are going to use for quantifying the distance between distributions.
\begin{definition}[spectral norm]
\label{def:spectralnorm}
The spectral or operator norm of a matrix $A \in \RL^{n \times d}$ is defined as \[\norm{A}=\sup\limits_{x\in\RL^d\setminus\{0\}} \frac{\norm{Ax}}{\norm{x}},\]
where $\norm{y}={(\sum_{i=1}^m y_i^2)}^\frac 1 2$ denotes the Euclidean vector norm for any $y\in \RL^m$.
\end{definition}

A useful fact that is straightforward from Definition \ref{def:spectralnorm} is that the spectral norm of a matrix $M$ equals its largest singular value, i.e., we have $\norm{M}=\sigma_{\rm max}(M)$ \citep[cf.][]{Horn1990}.

In order to quantify the distance between probability measures and in particular between the original posterior and its approximated counterpart we will need some further definitions. For this sake, given two probability measures $\gamma,\nu$ over $\RL^d$, let $\Lambda(\gamma,\nu)$ denote the set of all joint probability measures on $\RL^{d}\times \RL^{d}$ with marginals $\gamma$ and $\nu$, respectively.

\begin{definition}[Wasserstein distance, cf. \cite{Villani2009}]
\label{def:wassdist}
Given two probability measures $\gamma,\nu$ on $\RL^d$ the $\ell_2$ Wasserstein distance between $\gamma$ and $\nu$ is defined as
\[
\mathcal W_2(\gamma,\nu) =\left( \inf\limits_{\lambda\in\Lambda(\gamma,\nu)} \int_{\RL^{d}\times \RL^{d}} \norm{x-y}^2 \;\;\emph{d}\lambda(x,y) \right)^{\frac{1}{2}}
=\inf\limits_{\lambda\in\Lambda(\gamma,\nu)} \Ex[\lambda]{\norm{x-y}^2} ^{\frac{1}{2}}
\]
\end{definition}

From the definition of the Wasserstein distance we can derive a measure of how much points drawn from a given distribution will spread from the origin. The \emph{Wasserstein weight} can be thought of as a norm of a probability measure.

\begin{definition}[Wasserstein weight]
\label{def:wassweight}
We define the $\ell_2$ Wasserstein weight of a probability measure $\gamma$ as\\
\[
\mathcal W_2(\gamma)=\mathcal W_2(\gamma,\delta)\\
=\left( \int_{\RL^{d}} \norm{x}^2 \;\;\emph{d}\gamma \right)^{\frac{1}{2}}=\Ex[\gamma]{\norm{x}^2}^{\frac{1}{2}}
\]
where $\delta$ denotes the Dirac delta function.

\end{definition}

\subsection{Random projections and $\varepsilon$-subspace embeddings}
\label{sec:rapro:subspace}
The following definition of so called $\varepsilon$-subspace embeddings will be central to our work. Such an embedding can be used to reduce the size of a given data matrix while preserving the algebraic structure of its spanned subspace up to $(1\pm\varepsilon)$ distortion. Before we summarize several methods to construct a subspace embedding for a given input matrix, we give a formal definition. Here and in the rest of the paper we assume $0<\varepsilon\leq 1/2$.

\begin{definition}[$\varepsilon$-subspace embedding]
\label{def:subspace}
Given a matrix $U\in \RL^{n \times d}$ with orthonormal columns, an integer $k\leq n$ and an approximation parameter $0<\varepsilon\leq 1/2$, an $\varepsilon$-subspace embedding for $U$ is a map $\Pi: \RL^n \rightarrow \RL^k$ such that
\begin{eqnarray}
\label{def:subspace1}
(1-\varepsilon)\norm{Ux}^2\leq \norm{\Pi Ux}^2\leq (1+\varepsilon)\norm{Ux}^2
\end{eqnarray}
holds for all $x\in \RL^d$, or, equivalently
\begin{eqnarray}
\label{def:subspace2}
\norm{U^T\Pi^T\Pi U - I_d}\leq \varepsilon.
\end{eqnarray}
\end{definition}

Inequality (\ref{def:subspace1}) is mainly used in this paper, but Inequality (\ref{def:subspace2}) is more instructive in the sense that it makes clear that the embedded subspace is close to the identity, not involving much scale or rotation.

Note that an $\varepsilon$-subspace embedding $\Pi$ for the column\-space of a matrix $M$ preserves its squared singular values up to $(1\pm\varepsilon)$ distortion, which in particular means that it also preserves its rank. We prove this claim for completeness.

\begin{observation}
\label{obs:singular}
Let $\Pi$ be an $\varepsilon$-subspace embedding for the column\-space of $M\in\RL^{n\times d}$. Then $$(1-\varepsilon)\,\sigma_i^2(M) \leq \sigma_i^2(\Pi M) \leq (1+\varepsilon)\, \sigma_i^2(M)$$ and $$(1-2\varepsilon)\,\sigma_i^{-2}(M) \leq \sigma_i^{-2}(\Pi M) \leq (1+2\varepsilon)\,\sigma_i^{-2}(M).$$
\end{observation}
\begin{proof}
For the first claim, we make use of a min-max representation of the singular values that is known as the Courant-Fischer theorem \citep[cf.][]{Horn1990}. In the following derivation we choose $x^*$ to be the maximizer of (\ref{sing:1}) and $S^*$ the minimizer of (\ref{sing:2}).
\begin{eqnarray}
\sigma_i^2(\Pi M)&=&\min_{S\in\RL^{(i-1)\times d}}\max_{Sx=0,\norm{x}=1}\norm{\Pi Mx}^2\notag\\
&\leq & \max_{S^*x=0,\norm{x}=1}\norm{\Pi Mx}^2\label{sing:1}\\
&= & \norm{\Pi Mx^*}^2\notag\\
&\leq & (1+\varepsilon)\,\norm{Mx^*}^2\notag\\
&\leq & (1+\varepsilon)\,\max_{S^*x=0,\norm{x}=1}\norm{Mx}^2\notag\\
&= & (1+\varepsilon)\,\min_{S\in\RL^{(i-1)\times d}}\max_{Sx=0,\norm{x}=1}\norm{Mx}^2\label{sing:2}\\
&= & (1+\varepsilon)\,\sigma_i^2(M).\notag
\end{eqnarray}
The lower bound can be derived analogously using the lower bound of (\ref{def:subspace1}).

Now we use the first claim to prove the second. To this end, we bound the difference
\begin{eqnarray*}
\left|\frac{1}{\sigma_i^{2}(M)} - \frac{1}{\sigma_i^{2}(\Pi M)}\right| &=& \frac{|\sigma_i^{2}(M) - \sigma_i^{2}(\Pi M)|}{\sigma_i^{2}(M)\sigma_i^{2}(\Pi M)}
\leq \frac{\varepsilon \sigma_i^{2}(M) }{(1-\varepsilon)\,\sigma_i^{4}(M)}\\
&=& \frac{\varepsilon}{(1-\varepsilon)}\,\sigma_i^{-2}(M)
\leq 2\varepsilon\,\sigma_i^{-2}(M). 
\end{eqnarray*}
\end{proof}

There are several ways to construct an $\varepsilon$-subspace embedding. One of the more recent methods is using a so called \emph{graph-sparsifier}, which was initially introduced for the efficient construction of sparse sub-graphs with good expansion properties \citep{BatsonSS12}. The work of \cite{BoutsidisDM2013} adapted the technique to work for ordinary least-squares regression. While the initial construction was deterministic, they also gave alternative constructions combining the deterministic decision rules with non-uniform random sampling techniques. Another approach is subspace preserving sampling of rows from the data matrix. This technique was introduced by \cite{DrineasMM06} for $\ell_2$ regression and generalized to more general subspace sampling for the $p$-norm \citep{DasguptaDHKM09}. The sampling is done proportional to the so called \emph{statistical leverage scores}. These techniques have recently been analyzed and extended in a statistical setting as opposed to the algorithmic worst case setting \citep{RaskuttiM15, MaMY14}. All the aforementioned methods are in principle applicable whenever it is possible to read the input multiple times. For instance, one needs two passes over the data to perform the subspace sampling procedures, one for preprocessing the input matrix and another for computing the probabilities and for the actual sampling. 
This way one might reach a stronger reduction or better statistical properties since their (possibly random) construction depends on the input itself and therefore uses more information.

In principle our approximation results are independent of the actual method used to calculate the embedding as long as the property given in Definition \ref{def:subspace} is fulfilled. However, when the number of observations grows really massive or we deal with an infinite stream, then the data can only be read once, given time and space constraints. In order to use $\varepsilon$-subspace embeddings in a single-pass streaming algorithm, we consider the approach of so called \emph{oblivious} subspace embeddings in this paper. These can be viewed as distributions over appropriately structured $k\times n$ matrices from which we can draw a realization $\Pi$ independent of the input matrix. It is then guaranteed that for any fixed matrix $U$ as in Definition \ref{def:subspace} and failure probability $0<\alpha\leq 1/2$, $\Pi$ is an $\varepsilon$-subspace embedding with probability at least $1-\alpha$. The results of our work are always conditioned on the event that the map $\Pi$ is an $\varepsilon$-subspace embedding omitting to further mention the error probability $\alpha$. The reader should keep in mind that there is the aforementioned possibility of failure during the phase of sketching the data.

Instead of a subspace embedding, one might consider $\Pi = X^{T}$ a suitable choice for a
sketching matrix leading to a sketch of size $d\times (d+1)$ from which the exact likelihood and posterior can be characterized in some analytically tractable cases. However, it is well known that using the resulting matrix $X^T [X,Y]$ may result in numerical instabilities leading to bad conditioning of the covariance matrix $X^TX$. This effect can occur in the presence of collinearities or slight variations from orthogonality independent of the size of the data, and may lead to highly inaccurate frequentist estimators, \citep[cf.][]{LawsonH1995}. In a Bayesian setting, as we consider in this work, the instabilities result in extremely large variances of the MCMC sample, leading to simulations that do not converge. We will observe this behavior in Section \ref{sec:simstudy}.

We therefore consider the following approaches for obtaining oblivious $\varepsilon$-subspace embeddings:

1. \textbf{The Rademacher Matrix (RAD)}:
$\Pi$ is obtained by choosing each entry independently from $\{-1,1\}$ with equal probability. The matrix is then rescaled by $\frac{1}{\sqrt{k}}$. This method has been shown by \cite{Sarlos2006} to form an $\varepsilon$-subspace embedding with probability at least $1-\alpha$ when choosing essentially $k=O(\frac{d\log (d/\alpha)}{\varepsilon^2})$. This was later improved to $k=O(\frac{d+\log(1/\alpha)}{\varepsilon^2})$ in \cite{CW09}, which was recently shown to be optimal by \cite{NelsonN13lower}. While this method yields the best reduction among the different constructions that we consider in the present work, the RAD embedding has the disadvantage that we need $\Theta(ndk)$ time to apply it to an $n\times d$ matrix when streaming the input in general. If the input is given row by row or at least block by block, our implementation applies a fast matrix multiplication algorithm to each block. We remark that it is provably sufficient that the $\{-1,1\}$-entries in each row of the RAD matrix are basically four wise independent, i.e., when considering up to four entries of the same row, these behave as if they were fully independent. Such random numbers can be generated using a hashing scheme that generates BCH codes using a seed of size $O(\log n)$. This has first been noticed by \citet{AlonMS99}. In our implementation we have used the four wise independent BCH scheme as described in \citet{RusuD07}.

2. \textbf{The Subsampled Randomized Hadamard Transform (SRHT)} (originally from \cite{AilonL09}) is an embedding that is chosen to be $\Pi=RH_mD$ where $D$ is an $m\times m$ diagonal matrix where each entry is independently chosen from $\{-1,1\}$ with equal probability. The value of $m$ is assumed to be a power of two. It is convenient to choose the smallest such number that is not smaller than $n$. $H_m$ is the \emph{Hadamard-matrix} of order $m$ and $R$ is a $k\times m$ row sampling matrix. That is, each row of $R$ contains exactly one $1$-entry and is $0$ everywhere else. The index of the $1$-entry is chosen uniformly from $[m]$ i.i.d. for every row. The matrix is then rescaled by $\frac{1}{\sqrt{k}}$. Since $m$ is often larger than $n$, the input data must be padded with $0$-entries to compute the product $\Pi X$. Of course, it is not necessary to do that explicitly since all multiplications by zero can be omitted. The target dimension needed to form an $\varepsilon$-subspace embedding with probability at least $1-\alpha$ using this family of matrices was shown by \cite{BoutsidisG13} to be $k=O(\frac{(\sqrt{d}+\sqrt{\log n})^2 \log(d/\alpha)}{\varepsilon^2})$, which improved upon previous results from \cite{Drineas11}. Using this method, we have a small dependency on $n$, which is negligible whenever $n=O(\exp(d))$. This is often true in practice when $d$ is reasonably large. Compared to the RAD method, the dependency on the dimension $d$ is worse by essentially a factor of $O(\log d)$. It is known that $k=\Omega(d\log d)$ is necessary due to the sampling based approach. This was shown by \emph{reduction} from the coupon collectors problem, i.e., solving one problem can be \emph{reduced} to solving the other. See \cite{HalkoMT11} for details. The benefit that we get is that due to the inductive structure of the Hadamard matrix, the embedding can be applied in $O(nd\log k)$ time, which is considerably faster. It has been noticed in the original paper \citep{AilonL09} that the construction is closely related to four wise independent BCH codes. To our knowledge, there is no explicit proof that it is sufficient to use random bits of little independence. However, we use again the four wise BCH scheme for the implicit construction of the matrix $D$ and the linear congruency generator from the standard library of C++ 11 for the uniform subsampling matrix $R$. We will see in the empirical evaluation that this works well in practice.

3. \textbf{The Clarkson Woodruff (CW) sketch} \citep{ClarksonW13} is the most recent construction that we consider in this article. In this case the embedding is obtained as $\Pi=\Phi D$. The $n \times n$ matrix $D$ is constructed in the same way as the diagonal matrix in the SRHT case. Given a random map $h:[n]\rightarrow [k]$ such that for every $i\in[n]$ its image is chosen to be $h(i)=t\in[k]$ with probability $\frac{1}{k}$, again $\Phi$ is a binary matrix whose $1$-entries can be defined by $\Phi_{h(i),i}=1$. All other entries are $0$. This is obviously the fastest embedding, due to its sparse construction. It can be applied to any matrix $X\in\RL^{n\times d}$ in $O(\operatorname{nnz}(X))=O(nd)$ time, where $\operatorname{nnz}(X)$ denotes the number of non-zero entries in $X$. This is referred to as \emph{input sparsity time} and is clearly optimal up to small constants, since this is the time needed to actually read the input from a data stream or external memory, which dominates the sketching phase. However, its disadvantage is that the target dimension is $k=\Omega(d^2)$ \citep{NelsonN13lowerd2}. Roughly spoken, this is necessary due to the need to obliviously and perfectly hash $d$ of the standard basis vectors spanning $\RL^n$. Improved upper bounds over the original ones of \cite{ClarksonW13} show that $k=O(\frac{d^2}{\varepsilon^2\alpha})$ is sufficient to draw an $\varepsilon$-subspace embedding from this distribution of matrices with probability at least $1-\alpha$ \citep{NelsonN13}. This reference also shows that it is sufficient to use only four wise independent random bits to generate the diagonal matrix $D$. Again, in our implementation we use the four wise independent BCH scheme from \cite{RusuD07}. Moreover, $\Phi$ can be constructed using only pairwise independent entries. This can be achieved very efficiently using the fast universal hashing scheme introduced by \cite{DietzfelbingerHKP97} which we have employed in our implementation. The space requirement is only $O(\log n)$ for a hash function from this class. For a really fast implementation using bit-wise operations the actual size parameters of the sketch are chosen to be the smallest powers of two that are larger than the required sizes $n$ and $k$.

\begin{table}[tb]
\begin{center}
\begin{tabular}{lll}
\toprule
sketching method & target dimension & running time\\
\midrule
RAD   & $O\left(\frac{d+\log(1/\alpha)}{\varepsilon^2}\right)$ & $O\left(ndk\right)$\\
SRHT & $O\left(\frac{d \cdot \log(d/\alpha)}{\varepsilon^2}\right)$ & $O\left(nd\log k\right)$\\
CW    & $O\left(\frac{d^2}{\varepsilon^2\alpha}\right)$ & $O\left(\operatorname{nnz}(X)\right)=O(nd)$ \\
\bottomrule
\end{tabular}
\end{center}
\caption{Comparison of the three considered $\varepsilon$-subspace embeddings; $\operatorname{nnz}(X)$ denotes the number of non-zero entries in $X$, $\alpha$ denotes the failure probability.}
\label{tab:sketchcomparison}
\end{table}

Table \ref{tab:sketchcomparison} summarizes the above discussion, in particular the trade-off behavior between time and space complexity of the presented sketching methods. While in general one is interested in the fastest possible application time, memory constraints might make it impossible to apply the CW sketch due to its quadratic dependency on $d$. Taking it the other way, for a fixed sketching size, CW will give the weakest approximation guarantee \citep[cf.][]{YangMM15}. For really large $d$, even the $O(d\log d)$ factor of SRHT might be too large so that we have to use the slowest RAD sketching method.

\subsection{Extension to the streaming model}

The presented reduction techniques are of interest whenever we deal with medium to large sized data for reducing time and space requirements. However, when the data grows massive, we need to put more importance on the computational requirements. We therefore want to briefly discuss and give references to some of these technical details. For example, while the dimensions of the resulting sketches do not depend on $n$, this is not true for the embedding matrices $\Pi\in \RL^{k\times n}$. However, due to the structured constructions that we have surveyed above, we stress that the sketching matrices can be stored implicitly by using the different hash functions of limited independence. The hash functions used in our implementations are the four wise independent BCH scheme used in the seminal work of \cite{AlonMS99} and the universal hashing scheme by \cite{DietzfelbingerHKP97}. These can be evaluated very efficiently using bit-wise operations and can be stored using a seed whose size is only $O(\log n)$. Note that even this small dependency on $n$ is only needed in the sketching phase. After the sketch has been computed, the space requirements will be independent of $n$. A survey and evaluation of alternative hashing schemes can be found in \cite{RusuD07}.

The linearity of the embeddings allows for efficient application in sequential streaming and in distributed environments, see e.g. \cite{CW09,WoodruffZ13,KannanVW14}. The sketches can be updated in the most flexible dynamic setting, which is commonly referred to as the \emph{turnstile} model \citep{Muthukrishnan05}. In this model, think of an initial matrix of all zero values. The stream consists of updates of the form $(i,j,u)$ meaning that the entry $X_{ij}$ will be updated to $X_{ij}+u$. A single entry can be defined by one single update or by a sequence of not necessarily consecutive updates. For example a stream $S=\{\ldots,(i,j,+5),\ldots,(i,j,-3), \ldots\}$ will result in $X_{ij}=2$. Even deletions are possible in this setting by using negative updates. Clearly this also allows for additive updates of rows or columns, each consisting of consecutive single updates to all the entries in the same row or column.
At first sight this model might seem very technical and unnatural. But the usual form of storing data in a table is not appropriate or performant for massive data sets. The data is rather stored as a sequence of \emph{(key, value)} pairs in arbitrary order. For dealing with such unstructured data, the design of algorithms working in the turnstile model is of high importance. 

For distributed computations, note that the embedding matrices can be communicated efficiently to every machine in a computing cluster of $l$ machines. This is due to the small implicit representation by hash functions. Now, suppose the data is given as $X=\sum_{i=1}^l X^{(i)}$ where $X^{(i)}$ is stored on the machine with index $i\in [l]$. Note that by the above data representation in form of updates, $X^{(i)}$ can consist of rows, columns or single entries of $X$. Again, multiple updates to the same entry are possible and may be distributed to different machines. Every machine $i\in [l]$ can compute a small sketch on its own share $X^{(i)}$ of the data and efficiently communicate it to one dedicated central server. A sketch of the entire data set can be obtained by summing up the single sketches since $\Pi X = \sum_{i=1}^l \Pi X^{(i)}$. More details can be found in \cite{KannanVW14}. Recent implementations of similar distributed and parallel approaches for OLS are given by \cite{YangMM15}.

The above discussions make clear that our methods suit the criteria that need to be satisfied when dealing with Big Data \citep[cf.][]{welling:2014}. Namely, the number of data items that need to be accessed at a time is only a small subset of the whole data set, particularly independent of $n$. The algorithms should be amenable to distributed computing environments like MapReduce.


 \section{Theory}

In this section we introduce and develop the theoretical foundations of our approach and will combine them with existing results on ordinary least squares regression to bound the Wasserstein distance between the original likelihood function and its counterpart that is defined only on the considerably smaller sketch. Empirical evaluations supporting and complementing our theoretical results will be conducted in the subsequent section.

\subsection{Embedding the likelihood}

The following observation is standard (cf. \cite{givens1984,vempala2009}) and will be helpful in bounding the $\ell_2$ Wasserstein distance of two Gaussian measures. It allows us to derive such a bound by inspecting their means and their covariances separately.

\begin{observation}
\label{obs:meanvsvar}
Let $Z_1,Z_2\in \RL^d$ be random variables with finite first moments $m_1,m_2 < \infty$ and let $Z_1^m=Z_1-m_1$, respectively, $Z_2^m=Z_2-m_2$ be their mean-centered counterparts. Then it holds that $$\Ex{\norm{Z_1-Z_2}^2} = \norm{m_1-m_2}^2 + \Ex{\norm{Z_1^m-Z_2^m}^2}.$$
\end{observation}

\begin{proof}
\begin{eqnarray*}
\Ex{\norm{Z_1-Z_2}^2}
&=& \Ex{\norm{Z_1^m - Z_2^m + m_1 - m_2}^2} \\
&=& \Ex{\norm{Z_1^m - Z_2^m}^2 + \norm{m_1 - m_2}^2} ~ +\,2\;(m_1-m_2)^T\,\underbrace{\Ex{Z_1^m-Z_2^m}}_{=0}\\[-0.75\baselineskip]
&=& \Ex{\norm{Z_1^m - Z_2^m}^2} + \norm{m_1 - m_2}^2
\end{eqnarray*}
\end{proof}

In our first lemma we show that using an $ \varepsilon$-subspace embedding $\Pi$ for the column\-space of $[X,Y]$, we can approximate the least squares regression problem up to a factor of $1+\varepsilon$. That is, we can find a solution $\nu$ by projecting $\Pi Y$ into the column\-space of $\Pi X$ such that $\norm{X\nu-Y}\leq (1+\varepsilon) \min_{\beta\in \RL^d}\norm{X\beta-Y}$. Similar proofs can be found in \cite{fastcauchy2012} and \cite{BoutsidisDM2013}. We repeat the result here for completeness.

\begin{lemma}
\label{lem:l2regression}
Given $X\in\RL^{n\times d},Y\in\RL^{n}$, let $\Pi$ be an $(\varepsilon/3)$-sub\-space embedding for the column\-space of $[X,Y]$. Now let $\gamma=\operatorname{argmin}_{\beta\in\RL^d} \norm{X\beta-Y}^2$ and similarly define $\nu=\operatorname{argmin}_{\beta\in\RL^d} \norm{\Pi(X\beta-Y)}^2$. Then $$\norm{X\nu-Y}^2\leq (1+\varepsilon)\,\norm{X\gamma-Y}^2.$$
\end{lemma}
\begin{proof}
Let $[X,Y]=U\Sigma V^T$ denote the SVD of $[X,Y]$. Now define $\eta_1=\Sigma V^T [\gamma^T,-1]^T$ and $\eta_2=\Sigma V^T [\nu^T,-1]^T$. Using this notation we can rewrite $U\eta_1=X\gamma-Y$ and similarly $U\eta_2=X\nu-Y$. We have that\begin{eqnarray*}
(1-\varepsilon/3)\,\norm{U\eta_2}^2 \leq \norm{\Pi U\eta_2}^2
\leq \norm{\Pi U\eta_1}^2
\leq (1+\varepsilon/3)\,\norm{U\eta_1}^2.
\end{eqnarray*}
The first and the last inequality are direct applications of the subspace embedding property (\ref{def:subspace1}), whereas the middle inequality follows from the optimality of $\nu$ in the embedded subspace. Now, by rearranging and resubstituting terms, this yields \begin{eqnarray*}
\norm{X\nu-Y}^2 \leq \left( \frac{1+\varepsilon/3}{1-\varepsilon/3} \right)\norm{X\gamma-Y}^2
\leq (1+\varepsilon)\,\norm{X\gamma-Y}^2.
\end{eqnarray*} 
\end{proof}

One can even show that a distortion of order $\sqrt{\varepsilon}$, i.e., an $O(\sqrt{\varepsilon})$-subspace embedding is already enough to get the result. This was shown by using a more complicated proof taking the geometry of the least squares solution into account and using the property that the solution is obtained by an orthogonal projection onto the column\-space spanned by the data matrix \citep[cf.][]{CW09}. Putting it the other way around, by using an $(\varepsilon/3)$-subspace embedding as in Lemma \ref{lem:l2regression}, we even have\begin{equation}
\norm{X\nu-Y}^2\leq (1+\varepsilon^2)\,\norm{X\gamma-Y}^2. \label{eqn:epssq}
\end{equation}

In the following, we investigate the distributions proportional to the likelihood functions $p\propto \mathcal{L}(\beta | X,Y)$ and $p'\propto \mathcal{L}(\beta | \Pi X,\Pi Y)$ and bound their Wasserstein distance.

We begin our contribution with a bound on the distance of their means $\gamma$ and $\nu$, respectively. We generalize upon previous results for the specific embedding methods to arbitrary $\varepsilon$-subspace embeddings.

\begin{lemma}
\label{approx:means}
Given $X\in\RL^{n\times d},Y\in\RL^{n}$, let $\Pi$ be an $(\varepsilon/3)$-sub\-space embedding for the column\-space of $[X,Y]$. Now let $\gamma=\operatorname{argmin}_{\beta\in\RL^d} \norm{X\beta-Y}^2$ and similarly define $\nu=\operatorname{argmin}_{\beta\in\RL^d} \norm{\Pi(X\beta-Y)}^2$. Then $$\norm{\gamma - \nu}^2 \leq \frac{\varepsilon^2}{\sigma_{\emph{min}}^{2}(X)} \, \norm{X\gamma-Y}^2.$$
\end{lemma}
\begin{proof}
Let $X=U\Sigma V^T$ denote the SVD of $X$. Let $\eta=V^T(\gamma-\nu)$. First note that $\gamma$ and $\nu$ are both contained in the column\-space of $V$ \citep[cf.][]{Sarlos2006} which means that $V^T$ is a proper rotation with respect to $\gamma-\nu$. Thus,
\begin{eqnarray*}
\norm{X(\gamma-\nu)}^2 &=& \norm{U\Sigma V^T(\gamma-\nu)}^2
= \norm{\Sigma V^T(\gamma-\nu)}^2 \\
&=& \sum \sigma_i^2 \eta_i^2
\geq  \sum \sigma_{\text{min}}^2 \eta_i^2\\
&=& \sigma_{\text{min}}^2\,\norm{V^T(\gamma-\nu)}^2
= \sigma_{\text{min}}^2\,\norm{\gamma-\nu}^2\,.
\end{eqnarray*}

Consequently, it remains to bound $\norm{X(\gamma-\nu)}^2$. This can be done by using the fact that the minimizer $\gamma$ is obtained by projecting $Y$ orthogonally onto the column\-space of $X$. Therefore, we have $X^T(X\gamma-Y)=0$ \citep[cf.][]{CW09}. Furthermore, by Equation (\ref{eqn:epssq}) it holds that $\norm{X\nu - Y}^2 \leq (1+\varepsilon^2) \norm{X\gamma-Y}^2$. Now by plugging this into the Pythagorean theorem and rearranging we get that 
\begin{eqnarray*}
\norm{X(\gamma-\nu)}^2 &=& \norm{X\nu-Y}^2 - \norm{X\gamma-Y}^2
\leq \varepsilon^2\,\norm{X\gamma-Y}^2.
\end{eqnarray*}
Putting all together this yields the proposition
\begin{eqnarray*}
\norm{\gamma-\nu}^2 &\leq & \frac{1}{\sigma_{\text{min}}^{2}(X)}\,\norm{X(\gamma-\nu)}^2
\leq \frac{\varepsilon^2}{\sigma_{\text{min}}^{2}(X)}\,\norm{X\gamma-Y}^2.
\end{eqnarray*} 
\end{proof}

Now that we have derived a bound on the distance of the different means, recall that by Observation \ref{obs:meanvsvar}, we can assume w.l.o.g. $\gamma=\nu=0$ when we consider the variances. Namely, it remains to derive a bound on $\inf \Ex{\norm{Z_1^m-Z_2^m}^2}$, i.e., the least expected squared Euclidean distance of two points drawn from a joint distribution whose marginals are the \emph{mean-centered} original distribution and its embedded counterpart. Of course we can bound this quantity by explicitly defining a properly chosen joint distribution and bounding the expected squared distance for its particular choice. This is the idea that yields our next lemma.

\begin{lemma}
\label{approx:variance}
Let $p\propto \mathcal{L}(\beta |X,Y)$ and $p'\propto \mathcal{L}(\beta | \Pi X,\Pi Y)$. Let $Z_1^m,Z_2^m$ be the mean-centered versions of the random variables $Z_1 \sim p$ and $Z_2 \sim p'$ that are distributed according to $p$ and $p'$ respectively. Then we have
$$\inf\limits_{\lambda\in\Lambda(p,p')} \Ex[\lambda]{\norm{Z_1^m-Z_2^m}^2} \leq \varepsilon^2 \trace{(X^TX)^{-1}}.$$
\end{lemma}
\begin{proof}
Our plan is to design a joint distribution that deterministically maps points from one distribution to another in such a way that we can bound the distance of every pair of points. This can be done by utilizing the Dirac delta function $\delta(\cdot)$, which is a degenerate probability density function that concentrates all probability mass at zero and has zero density otherwise. Given a bijection $g: \RL^d \rightarrow \RL^d$ we can define such a joint distribution $\lambda\in \Lambda(p,p')$ through its conditional distributions $\lambda(x\,|\,y)=\delta(x-g(y))$ for every $y\in \RL^d$. It therefore remains to define $g$.

According to Observation \ref{obs:singular}, when applying the embedding $\Pi$, the column\-space of a matrix is expanded or contracted, respectively, by a factor of at most $(1\pm\varepsilon)$. We will make use of this fact in the following way. Let $X=U\Sigma V^T$ and $\Pi X=\tilde{U}\tilde{\Sigma}\tilde{V}^T$ denote the SVDs of $X$ and $\Pi X$, respectively. Now, to define the $x$-$y$-pairs that will be mapped to each other by $g$, we consider vectors $x,x',y,y' \in \RL^d$ where $x'$ and $y'$ are contained in the columnspaces of $V$ and $\tilde{V}$, respectively. To obtain the bijection $g$, let the vectors have the following properties for arbitrary, but fixed radius $c \geq 0$:
\begin{enumerate}
	\item $\norm{x'}=\norm{y'}=c$
	\item $x=\Sigma V^Tx'$
	\item $y=\tilde{\Sigma} \tilde{V}^Ty'$
	\item $\exists \tau > 0: x=\tau y$.
\end{enumerate}
Observe that by the first property $x'$ and $y'$ lie on a $d$-di\-men\-sion\-al sphere with radius $c$ centered at $0$. Therefore, there exists a rotation matrix $R\in\RL^{d\times d}$ such that $y'=Rx'$. Note that such a map is bijective by definition. The second item defines a map of such spheres to ellipsoids (also centered at $0$) given by $\Sigma V^T$. Recall that $x'$ was chosen from the column\-space of $V$. Thus, this map can be seen as bijection between the $d$-dimensional vector space and a $d$-dimensional subspace contained in $n$-dimensional space. The third property is defined analogously. The fourth property urges that $x$ and $y$ both lie on a ray emanating from $0$. Note that any such ray intersects each ellipsoid exactly once.

Our bijection can be defined accordingly as 
\begin{eqnarray*}
g: \RL^d &\rightarrow & \RL^d\\
x \;\; &\mapsto & \tilde{\Sigma} \tilde{V}^TRV\Sigma^{-1} x
\end{eqnarray*}
by composing the map $\Sigma V^T$, defined in the second item, with the rotation $R$ and finally with $\tilde{\Sigma} \tilde{V}^T$ from the third property. The map is bijective since it is obtained as the composition of bijections.

Now, in order to bound the distance $\norm{Z_1^m-Z_2^m}^2$ for any realization of $(Z_1^m,Z_2^m)$ according to their joint distribution defined above, we can derive a bound on the parameter $\tau$. Substituting the second and third properties into the fourth, we get that
\begin{eqnarray*}
\Sigma V^Tx' = \tau \tilde{\Sigma} \tilde{V}^Ty'
\end{eqnarray*}
which can be rearranged to
\begin{eqnarray*}
y'^Ty'\tau &=& (y'^T \tilde{V}) \tilde{\Sigma}^{-1} \Sigma (V^T x')\\
  ~  &=& \sum (y'^T \tilde{V})_i (V^Tx')_i \frac{\sigma_i}{\tilde{\sigma}_i}\\
  ~&\leq & \sum (y'^T \tilde{V})_i (V^Tx')_i \frac{\sigma_i}{\sigma_i\sqrt{1-\varepsilon}}\\
  ~&\leq & (1+\varepsilon)\sum (y'^T \tilde{V})_i (V^Tx')_i\\
  ~&\leq & (1+\varepsilon)\,c^2.
\end{eqnarray*}
The first inequality follows from $\tilde{\sigma}_i \geq \sqrt{1-\varepsilon} \; \sigma_i$ and the second from the assumption $\varepsilon \leq 1/2$. This eventually means that $\tau \leq (1+\varepsilon)$ since $y'^Ty'=c^2$ by the first property.

A lower bound of $\tau \geq (1-\varepsilon)$ can be derived analogously by using $\tilde{\sigma}_i \leq \sqrt{1+\varepsilon} \; \sigma_i$.

Now we can conclude our proof. It follows that
\begin{eqnarray*}
\inf\limits_{\lambda'\in \Lambda(p,p')} \Ex[\lambda']{\norm{Z_1^m-Z_2^m}^2} &\leq & \Ex[\lambda]{\norm{Z_1^m-Z_2^m}^2}
\leq \Ex[\lambda]{\norm{\varepsilon Z_1^m}^2}\\
~ &=& \varepsilon^2 \; \Ex[\lambda]{\norm{Z_1^m}^2}
  = \varepsilon^2 \, \trace{(X^TX)^{-1}}.
\end{eqnarray*}
The last equality holds since the expected squared norm of the mean-centered random variable is just the trace of its covariance matrix. 
\end{proof}
Combining the above results we get the following lemma.

\begin{lemma}
\label{thm:likelihood}
Let $\Pi$ be an $(\varepsilon/3)$-subspace embedding for the column\-space of $X$. Let $p\propto \mathcal{L}(\beta |X,Y)$ and $p'\propto \mathcal{L}(\beta | \Pi X,\Pi Y)$. Then
$$
\mathcal W_2^2(p,p') \leq \frac{\varepsilon^2}{\sigma_{\emph{min}}^{2}(X)} \, \norm{X\gamma-Y}^2 + \varepsilon^2 \trace{(X^TX)^{-1}}.
$$
\end{lemma}
\begin{proof}
The lemma follows from Definition \ref{def:wassdist}, Observation \ref{obs:meanvsvar}, Lemma \ref{approx:means} and Lemma \ref{approx:variance}. 
\end{proof}

Under mild assumptions we can argue that this leads to a $(1+O(\varepsilon))$-approximation of the likelihood with respect to the Wasserstein weight (see Definition \ref{def:wassweight}).

\begin{corollary}
\label{corr:wassweight}
Let $\Pi$ be an $(\varepsilon/3)$-subspace embedding for the column\-space of $X$. Let $p\propto \mathcal{L}(\beta |X,Y)$ and similarly let $p'\propto \mathcal{L}(\beta | \Pi X,\Pi Y)$. Let $\kappa(X)=\sigma_{\max}(X)/\sigma_{\min}(X)$ be the condition number of $X$. Assume that for some $\rho\in(0,1]$ we have $\norm{X\gamma} \geq \rho\norm{Y}$. Then $$\mathcal W_2(p') \leq \left(1+\frac{\kappa(X)}{\rho} \,{\varepsilon} \right)\;\mathcal W_2(p).$$
\end{corollary}
\begin{proof}
By definition, the squared $\ell_2$ Wasserstein weight of $p$ equals its second moment. Since $p$ is a Gaussian measure with mean $\gamma$ and covariance matrix $(X^TX)^{-1}$, we thus have\begin{eqnarray*} \mathcal W_2^2(p)=\norm{\gamma}^2 + \trace{(X^TX)^{-1}}
\end{eqnarray*}
and similarly
\begin{eqnarray*}
\mathcal W_2^2(p')=\norm{\nu}^2 + \trace{(X^T\Pi^T\Pi X)^{-1}}.
\end{eqnarray*}

Since $\Pi$ is an $(\varepsilon/3)$-subspace embedding for the column\-space of $X$ we know from Observation \ref{obs:singular}, that all the squared singular values of $X$ are approximated up to less than $(1\pm\varepsilon)$ error and so are their inverses. Therefore, we have that \begin{eqnarray}\trace{(X^T\Pi^T\Pi X)^{-1}}\leq (1+\varepsilon)\, \trace{(X^TX)^{-1}}.\label{eq:trace}\end{eqnarray}
It remains to bound $\norm{\nu}^2$. To this end we use the assumption that for some $\rho\in(0,1]$ we have $\norm{X\gamma} \geq \rho\norm{Y}$. By $X^T(X\gamma-Y)=0$ and applying the Pythagorean Theorem this means that
\begin{eqnarray}
\norm{X\gamma-Y}^2 &=& \norm{Y}^2 - \norm{X\gamma}^2 \notag
\leq \norm{X\gamma}^2 \left(\frac{1}{\rho^2}-1\right) \notag
\leq \frac{\norm{X\gamma}^2}{\rho^2}.\label{ineq:assumption}
\end{eqnarray}

Now we can apply the triangle inequality, Lemma \ref{approx:means}, Inequality (\ref{ineq:assumption}) and Definition \ref{def:spectralnorm} to get
\begin{eqnarray*}
\norm{\nu} &\leq & \norm{\gamma} + \norm{\nu-\gamma}\\
&\leq & \norm{\gamma} + \frac{\varepsilon}{\sigma_{\text{min}}(X)} \, \norm{X\gamma-Y}\\
&\leq & \norm{\gamma} + \frac{\varepsilon}{\rho\sigma_{\text{min}}(X)} \, \norm{X\gamma}\\
&\leq & \norm{\gamma} + \frac{\varepsilon}{\rho\sigma_{\text{min}}(X)} \, \norm{X}\norm{\gamma}\\
&=& \norm{\gamma} + \frac{\varepsilon}{\rho} \, \kappa(X)\norm{\gamma}\\
&=& \left(1+\frac{\kappa(X)}{\rho} \,\varepsilon \right)\norm{\gamma}.
\end{eqnarray*}
Combining this with Inequality (\ref{eq:trace}), the claim follows since $\frac{\kappa(X)}{\rho}\geq 1$ and therefore $(1+\varepsilon)\leq (1+\frac{\kappa(X)}{\rho}\varepsilon)^2$ and finally taking square roots on both sides. 
\end{proof}

We stress that the assumption that there exists some constant $\rho\in(0,1]$ such that $\norm{X\gamma} \geq \rho\norm{Y}$ is very natural and mild in the setting of linear regression since it means that at least a constant fraction of the dependent variable $Y$ can be explained within the column\-space of the data $X$ \citep[cf.][]{DrineasMM06}. If this is not true, then a linear model is not appropriate at all for the given data.

\subsection{Bayesian regression}

So far we have shown that using subspace embeddings to compress a given data set for regression yields a good approximation to the likelihood. Note that in a Bayesian regression setting Lemma \ref{thm:likelihood} already implies a similar approximation error for the posterior distribution if the priors for $\beta$ are chosen to be uniform distributions over $\mathds R$. This is an improper, non-informative choice, $p_{\rm pre}(\beta) = \mathds{1}_{\mathds{R}^d}$. From this, it follows that
\begin{eqnarray*}
p_{\rm post}(\beta|X, Y) & \propto & \mathcal L(\beta | X,Y)\cdot \mathds{1}_{\mathds{R}^d}\\
 & = & \mathcal L(\beta | X,Y).
\end{eqnarray*}
 The remaining term is just the Gaussian likelihood which is proper. For regression models, especially on data sets with large $n$, this covers a considerable amount of the cases of interest \citep[cf.][]{Gelman2014}. We will extend this to arbitrary Gaussian priors $p_{\rm pre}(\beta)$ leading to our main result: an approximation guarantee for Gaussian Bayesian regression in its most general form.

To this end, let $m$ be the mean of the prior distribution and let $S$ be derived from its covariance matrix by $\Sigma=\varsigma^2(S^TS)^{-1}$. Now, the posterior distribution is given by 
\begin{eqnarray*}
p_{\rm post}(\beta|X, Y) &\propto & \mathcal L(\beta | X,Y)\cdot p_{\rm pre}(\beta)\\
~ &=& \frac{1}{(2\pi\varsigma^2)^{n/2}} \cdot \exp\left(-\frac{1}{2\varsigma^2}\norm{X\beta-Y}^2\right)
\cdot \frac{1}{(2\pi)^{\frac{d}{2}}|\Sigma |^{\frac{1}{2}}} \cdot \exp\left(-\frac{1}{2\varsigma^2}\norm{S(\beta-m)}^2\right).
\end{eqnarray*} 
Thus, we know that up to some constants that are independent of $\beta$, the exponent of the posterior can be described by\begin{equation}
\norm{X\beta-Y}^2 + \norm{S(\beta-m)}^2 \label{eq:above}
\end{equation}
which contains all the information to define the mean and covariance structure of the posterior distribution.
Now let $$Z=\left[\begin{matrix}X\\S\end{matrix}\right] \quad\text{and}\quad z=\left[\begin{matrix}Y\\Sm\end{matrix}\right].$$
With these definitions we can rewrite Equation (\ref{eq:above}) above as $\norm{Z\beta-z}^2$. This, in turn, can be treated as a (frequentist) regression problem to which we can apply Lemma \ref{thm:likelihood}. We just have to use a subspace embedding for the column\-space of $[Z,z]$ instead of only embedding $[X,Y]$. We will see that it is not necessary to do this explicitly. More precisely, embedding only the data matrix is sufficient to have a subspace embedding for the entire column\-space defined by the data and the prior information and, therefore, to have a proper approximation of the posterior distribution. This is formalized in the following lemma.

\begin{lemma}
\label{lem:onlydata}
Let $M=[M_1^T,M_2^T]^T\in\RL^{(n_1+n_2)\times d}$ be an arbitrary matrix. Suppose $\Pi$ is an $\varepsilon$-subspace embedding for the column\-space of $M_1$. Let $I_{n_2}\in\RL^{(n_2\times n_2)}$ be the identity matrix.  Then 
\[
P=\left[\begin{matrix}\Pi & 0\\0 & I_{n_2}\end{matrix}\right]\in\RL^{(k+n_2)\times (n_1+n_2)}
\] is an $\varepsilon$-subspace embedding for the column\-space of $M$.
\end{lemma}
\begin{proof}
Fix an arbitrary $x\in\RL^d$. We have
\begin{eqnarray*}
|\norm{PMx}^2-\norm{Mx}^2|
&=& |\norm{\Pi M_1x}^2+\norm{M_2x}^2-\norm{M_1x}^2-\norm{M_2x}^2|\\
&=& |\norm{\Pi M_1x}^2-\norm{M_1x}^2|\\
&\leq & \varepsilon \norm{M_1x}^2\\
&\leq & \varepsilon (\norm{M_1x}^2 + \norm{M_2x}^2)\\
&=& \varepsilon \norm{Mx}^2
\end{eqnarray*}
which concludes the proof by singular value decomposition $M=U\Sigma V^T$ and surjectivity of the linear map $\Sigma V^T$. 
\end{proof}

This lemma finally enables us to prove our main theoretical result.

\begin{theorem}
\label{thm:posterior}
Let $\Pi$ be an $(\varepsilon/3)$-subspace embedding for the column\-space of $X$. Let $p_{\rm pre}(\beta)$ be an arbitrary Gaussian distribution with mean $m$ and covariance matrix $\Sigma=\varsigma^2(S^TS)^{-1}$. Let $$Z=\left[\begin{matrix}X\\S\end{matrix}\right] \quad\text{and}\quad z=\left[\begin{matrix}Y\\Sm\end{matrix}\right].$$ Let $\mu = \operatorname{argmin}_{\beta\in\RL^d} \norm{Z\beta-z}$ be the posterior mean. Let $p\propto \mathcal{L}(\beta |X,Y)\cdot p_{\rm pre}(\beta)$ and $p'\propto \mathcal{L}(\beta | \Pi X,\Pi Y)\cdot p_{\rm pre}(\beta)$. Then
$$
\mathcal W_2^2(p,p') \leq \frac{\varepsilon^2}{\sigma_{\emph{min}}^{2}(Z)} \, \norm{Z\mu-z}^2 + \varepsilon^2 \trace{(Z^TZ)^{-1}}.
$$
\end{theorem}
\begin{proof}
From our previous reasoning we know that approximating the posterior distribution can be reduced to approximating a likelihood function that is defined in terms of the data as well as the parameters of the prior distribution. This has been shown by rewriting Equation (\ref{eq:above}) above as $\norm{Z\beta-z}^2$. For that reason, we can apply Lemma \ref{thm:likelihood} to get the desired result if we are given an $(\varepsilon/3)$-subspace embedding for the column\-space of $Z$. Using Lemma \ref{lem:onlydata} we know that for this, it is sufficient to use an $(\varepsilon/3)$-subspace embedding for the column\-space of $[X,Y]$ independent of the covariance and mean that define the prior distribution. 
\end{proof}

Similar to Corollary \ref{corr:wassweight} we have the following result concerning the posterior distribution.

\begin{corollary}
\label{corr:wassweight2}
Let $\Pi$ be an $(\varepsilon/3)$-subspace embedding for the column\-space of $X$. Let $p_{\rm pre}(\beta)$ be an arbitrary Gaussian distribution with mean $m$ and covariance matrix $\Sigma=\varsigma^2(S^TS)^{-1}$. Let $$Z=\left[\begin{matrix}X\\S\end{matrix}\right] \quad\text{and}\quad z=\left[\begin{matrix}Y\\Sm\end{matrix}\right].$$ Let $\mu = \operatorname{argmin}_{\beta\in\RL^d} \norm{Z\beta-z}$ be the posterior mean. Let $p\propto \mathcal{L}(\beta |X,Y)\cdot p_{\rm pre}(\beta)$ and $p'\propto \mathcal{L}(\beta | \Pi X,\Pi Y)\cdot p_{\rm pre}(\beta)$. Let $\kappa(Z)$ be the condition number of $Z$. Assume that for some $\rho\in(0,1]$ we have $\norm{Z\mu} \geq \rho\norm{z}$. Then we have $$\mathcal W_2(p') \leq \left(1+\frac{\kappa(Z)}{\rho} \,\varepsilon \right)\;\mathcal W_2(p).$$
\end{corollary}

Both Theorem \ref{thm:posterior} and Corollary \ref{corr:wassweight2} show that the sketch preserves the expected value and the covariance structure of the posterior distribution very well. Note that for normal distributions, these parameters fully characterize the distribution as they are sufficient statistics. Therefore, one can see the corresponding parameters based on the sketched data set as very accurate approximate sufficient statistics for the posterior distribution.


\section{Simulation Study}
\label{sec:simstudy}

To validate the proposed method empirically, we conduct a simulation study. For this, we employ MCMC methods to obtain the posterior distributions for the parameters of the Bayesian regressions. Please note that the sketching techniques can also be combined with other methods. We concentrate on MCMC methods, however, as they are very reliable, widely-used and allow for easy checking of convergence. The different sketching methods were implemented as described in Section \ref{sec:rapro:subspace} and technically more detailed in the given references. All codes were written in the C++ 11 programming language and compiled using GCC 4.7. For fast matrix multiplications we employed the LAPACK 3.5.0 library where applicable. Our R-package \texttt{RaProR} uses the above implementations. The package is available on its project website \citep{RaProR:2015}. All simulations were done using R, version 3.1.2 \citep{r:2014} and the R-package \texttt{rstan}, version 2.3 \citep{rstan:2013}. The simulations were conducted on an Intel Xeon E5430 quad-core CPU running at 2.66 GHz using 16 GB DDR2 memory on a Debian GNU Linux 7.8 distribution. The hard drive used was a Seagate Momentus 7200.4 G-Force 500 GB, SATA 3Gb/s HDD with 7200 rpm and 16 MB cache.

\subsection{Data generation}
\label{subsec:simstudy:datagen}

For the simulation study, we create a set of data sets. We vary the number of observations $n$, the number of variables $d$, and the standard deviation of the error term $\varsigma$. The variation of $n$ is introduced to monitor whether the running time of the analyses based on sketches is indeed independent of $n$ and also to see how well the proposed method deals with growing data sets. We choose values of $n \in \{50\,000, 100\,000, 500\,000, 1\,000\,000\}$. The size of the sketches depends mainly on the number of variables in the data set. For this reason, we conduct simulations with two values of $d$, $d=50$ and $d=100$. The reason for choosing rather small values of $n$ and $d$ is that our aim is to compare the results of the MCMC on the sketched data sets to the results on the respective original data set. The sketching methods presented here can handle larger values of $d$ and arbitrary values of $n$, but employing MCMC on the original data set then becomes unfeasible. The standard deviation of the error term is important, because the goodness of the approximation also depends on the goodness of the model (cf. Lemma \ref{thm:likelihood} and Theorem \ref{thm:posterior}). Here, we choose $\varsigma \in \{1, 2, 5, 10\}$, thus ranging from very well-fitting models to models with quite high error variance. 

The generated true values of $\beta$ follow a zero-inflated Poisson distribution, where the expected value of the Poisson distribution is $3$ and the probability of a component exhibiting an excess zero is equal to $0.5$. This means that the components of $\beta$ have no influence, i.e. are 0, with probability $0.5$ and follow a $\text{Poi}(3)$ distribution with probability $0.5$. All components that follow a $\text{Poi}(3)$ distribution are multiplied with $(-1)$ with probability $0.5$. The data set $X$ is obtained in two steps. At first, a $d$-dimensional vector that represents the column means is drawn randomly from a $N(0, 25)$ distribution. In a second step, the actual values of $X$ are drawn from a normal distribution with the column mean as expected value and variance of four. The variance in the columns of $X$ is thus lower than the error variance for two of our choices and the same or less for the other two choices. $Y$ is then generated by multiplying $X$ with $\beta$ and adding the error term, in accordance with the model.

\subsection{Regression model}
\label{subsec:simstudy:model}

We employ a standard Bayesian linear regression model (cf. Section \ref{sec:prelim:bayesreg})
\begin{equation*}
Y \sim N(X\beta, \varsigma^2 I_n)
\end{equation*}
with independent uniform priors over $\mathds{R}$ for all components of $\beta$, which are improper, non-informative prior distributions. For $\varsigma$, the uniform prior is limited to the positive part of $\mathds{R}$. We choose an improper uniform prior rather than an inverse gamma prior with small values for the hyperparameters as \cite{Gelman2006} indicates that such priors can have a skewing effect on the posterior distribution. When using improper prior distributions, it is necessary to ensure that the posterior distributions are proper. For our choice of uniform distributions, this does not pose a problem, since the uniform priors are represented by indicators, which are constant over $\mathds{R}$, $p_{\rm pre}(\beta) = \mathds{1}_{\mathds{R}^d}$, and for that reason do not influence the likelihood. Thus, it follows that
\begin{eqnarray*}
p_{\rm post}(\beta|X, Y) & \propto & \mathcal L(\beta | X,Y)\cdot \mathds{1}_{\mathds{R}^d}\\
 & = & \mathcal L(\beta | X,Y).
\end{eqnarray*}
 The remaining term is just the Gaussian likelihood which is proper with respect to both, $\beta$ and $\varsigma$ \citep[cf.][]{Gelman2014}. Although closed-form expressions are known for this model, we employ MCMC for the reasons motivated in Sections \ref{sec:background} and \ref{sec:prelim:bayesreg}.

Our theoretical guarantees comprise the posterior distributions of $\beta$. Our model is thus more general than the theoretical results. As an alternative, $\varsigma$ can be fixed to an estimated value obtained using $\|\Pi X \beta - \Pi Y\|_2/\sqrt{n}$. The results we present in the following are all based on $\beta$.

\subsection{Preliminary simulations}
\label{subsec:simstudy:prelim}

In a first step, we consider the running times necessary to carry out Bayesian regression for (non-embedded) data sets with an increasing number of observations $n$, employing the No-U-Turn-Sampler \citep{hoffman-gelman:2014}, which is implemented in the R-package \texttt{rstan}. We use standard settings and four chains, which are sampled in parallel. The resulting running times are plotted in Figure \ref{fig:preruntimed50}. The running time depends at least linearly on $n$, with occasional jumps that are probably due to swaps to external memory, which slows down the computation by a large factor. There is an outlier for the case of $n=50$, which takes a lot more time to compute than $n=100$ and even more than $n=200$. Since we have $d=50$ variables, the total number of parameters (52) is higher than the number of observations for this case only. While the linear dependency on the number of observations does not pose a problem for small to medium-sized data sets, for big data settings MCMC methods become unfeasible. This underlines the usefulness of embedded data sets.

\begin{figure}[tb]
\begin{center}
\includegraphics[width = 0.75\textwidth]{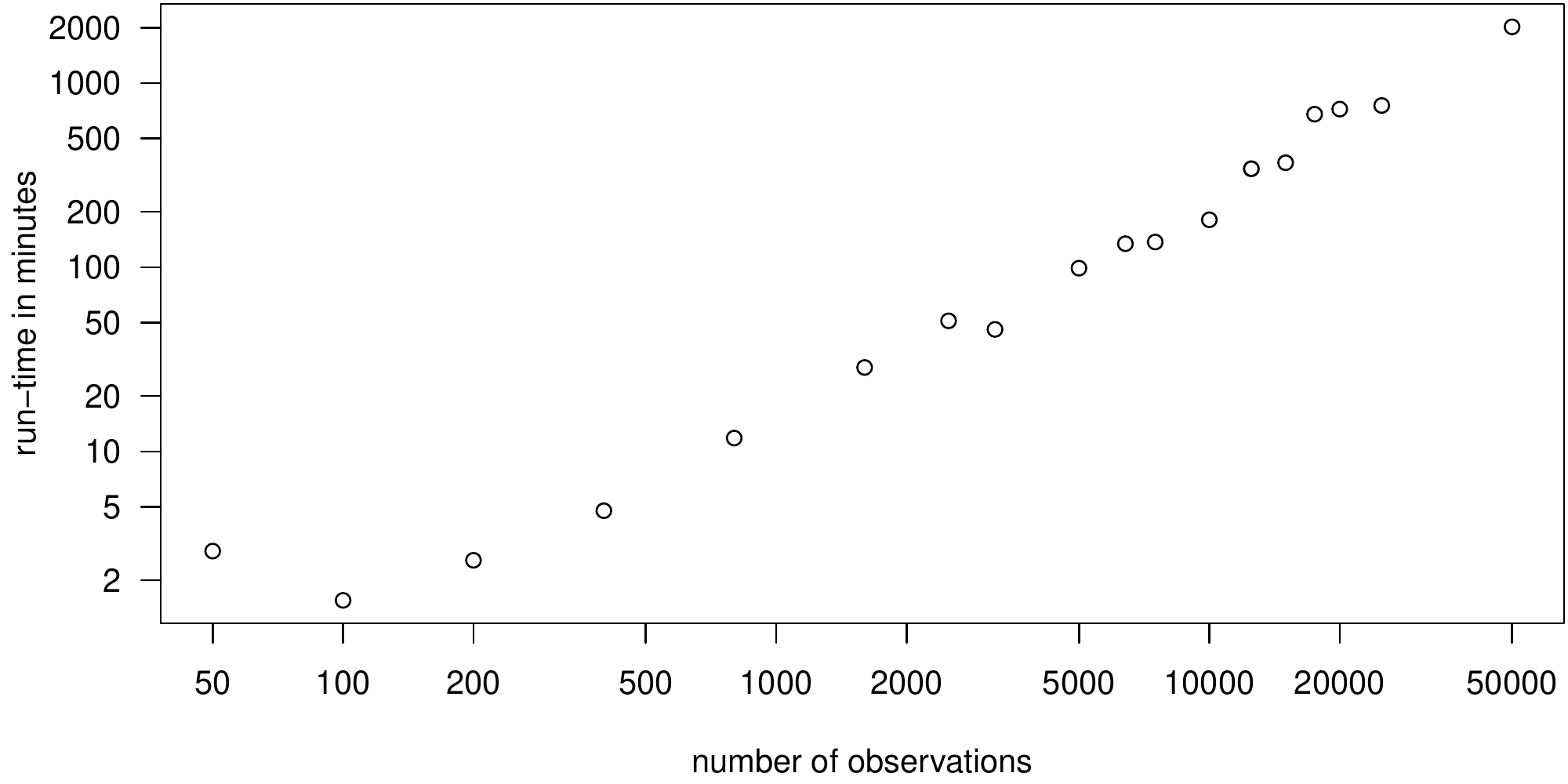}
\end{center}
\caption{Running times for simulated data sets with varying number of observations and $d=50$ variables}
\label{fig:preruntimed50}
\end{figure}

Before we consider our embedding methods, we pick up on the idea of using $\Pi = X^T$ as sketching matrix. If the data set is given row by row, the sketch $\Pi X = X^T X$ can be computed efficiently using the tensor product of the row vectors $X^T X = \sum x_i^T x_i$, which is a sum of $d\times d$ matrices with rank one. Analoguously, we have $\Pi Y = X^T Y = \sum x_i^T Y_i$. This results in a sketch of size $d \times (d+1)$, which is the smallest possible sketch for problems of full rank and has no error at all in the sense that the exact likelihood respectively posterior distribution can be calculated from these matrices. We have argued before that this method is not numerically stable in general. As we focus on MCMC in this work, we show how this effect influences the run of the MCMC sampler. We tried analyzing Bayesian linear regression models based on $X^T [X, Y]$, using some of the data sets described in Section \ref{subsec:simstudy:datagen} and also a smaller data set with $n=10\,000$, which was generated in the same way. We have found that the models do not converge in practice. Increasing the number of iterations does not seem to be a remedy as the variance of the MCMC sample grows with more iterations.  Figure \ref{fig:transdiv} shows an exemplary traceplot for one parameter, consisting of four chains. The range of the sample is enormous. The variation is high for all of the chains, they exhibit standard deviations in the order of $10^{9}$. When reducing the number of iterations from the $10\,000$ in Figure \ref{fig:transdiv} to, say, $5\,000$, the range of the MCMC sample decreases markedly. However, there is no sign of convergence with minimum and maximum around $-4\cdot 10^{9}$ and $4\cdot 10^{9}$, respectively.

\begin{figure}[tb]
\begin{center}
\includegraphics[width = 0.75\textwidth]{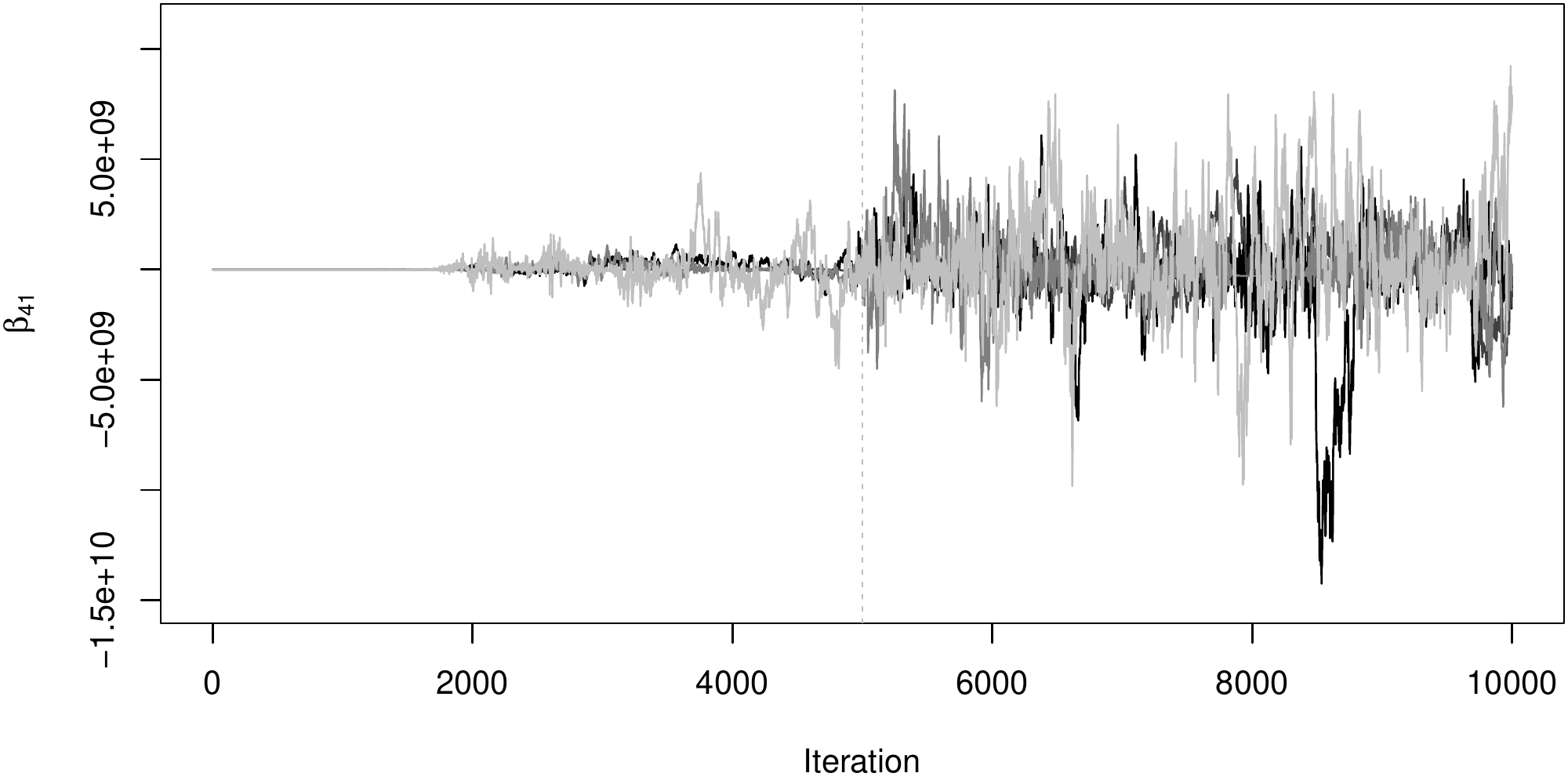}
\end{center}
\caption{Traceplot of MCMC sample (with four chains) for one parameter based on data set obtained using $\Pi = X^T$ as sketching matrix. Original data set contains $n=10\,000$ observations and $d = 50$ variables}
\label{fig:transdiv}
\end{figure}

We can deal with both of these issues using subspace embeddings as we can underline in our next experiments. We conduct a series of simulations that aims at comparing our proposed method to the standard meth\-od on the original data. To obtain the subspace embeddings, we employ the three approaches described in Section \ref{sec:prelim}. As approximation parameter, we choose $\varepsilon = 0.1$ and $\varepsilon =0.2$ for all three methods. We do not recommend using values of $\varepsilon>0.2$. Table \ref{tab:sketchsize} contains the number of observations of the sketches depending on the number of variables and the value of the approximation parameters. The sizes for RAD and SRHT are both set to $k=\lceil\frac{d\log d}{\varepsilon^2}\rceil$ to be comparable. They differ by one due to rounding errors. For the CW sketch we used $k$ equal to the smallest power of two larger than $\frac{d^2}{20\varepsilon^2}$. Please note that the CW embeddings generally result in a higher number of observations due to the quadratic dependency on $d$. However, the opposite is true for $50$ variables. This is due to the constants that we used. The constants are set to $1$ in the case of RAD and SRHT. For the RAD method this was empirically evaluated by \cite{SureshW11}. For CW the constant may be much smaller as indicated by a lower bound in \cite{NelsonN13lowerd2}. Preliminary experiments on small scale led to our choice of $\frac{1}{20}$.


\begin{table}[tb]
\begin{center}
\begin{tabular}{rrrrr}
\toprule
$d$ & $\varepsilon$ & RAD & SRHT & CW\\
\midrule
50 & 0.1 & $20\,546$ & $20\,547$ & $16\,384$\\
50 & 0.2 & $5\,136$ & $5\,137$ & $4\,096$\\
\midrule
100 & 0.1 & $47\,174$ & $47\,175$ & $65\,536$\\
100 & 0.2 & $11\,793$ & $11\,794$ & $16\,384$\\
\bottomrule
\end{tabular}
\end{center}
\caption{Number of observations of the sketches for different values of $d$ and $\varepsilon$}
\label{tab:sketchsize}
\end{table}%

\subsection{Comparison of posterior means}
\label{sec:simstudy:postmeans}

To evaluate the results, we first compare the posterior means of the models based on the embedded data sets with the posterior means of the model based on the original data set. Table \ref{tab:absdevfulld50} contains an overview of the sum of squared distances between the embedded data sets' posterior means and those of the original model. Geometrically, this is the squared Euclidean distance of the posterior mean vectors. As indicated by Theorem \ref{thm:posterior}, the sum of squared distances grows with the standard deviation of the error term. There does not seem to be a systematic difference in performance between the different sketching methods. With larger $\varepsilon$, we usually, but not necessarily observe an increase in the distance. Please note that some values are missing, because the original models did not converge within reasonable time bounds.


\begin{table}[tb]
\begin{center}
\begin{tabular}{lllrrrr}
  \toprule
n & sketch & $\varepsilon$ & $\varsigma=1$ & $\varsigma=2$ & $\varsigma=5$ & $\varsigma=10$ \\ 
  \midrule
  $5\cdot 10^4$ & RAD & 0.1 & 0.052 & 0.025 & 0.021 & 0.834 \\ 
  $5\cdot 10^4$ & RAD & 0.2 & 0.014 & 0.781 & 0.892 & 1.512 \\ 
  $5\cdot 10^4$ & SRHT & 0.1 & 0.001 & 0.009 & 0.021 & 0.165 \\ 
  $5\cdot 10^4$ & SRHT & 0.2 & 0.004 & 0.077 & 0.093 & 0.757 \\ 
  $5\cdot 10^4$ & CW & 0.1 & 0.025 & 0.004 & 0.021 & 0.195 \\ 
  $5\cdot 10^4$ & CW & 0.2 & 0.016 & 0.040 & 0.156 & 0.915 \\ 
  $1\cdot 10^5$ & RAD & 0.1 &  &  & 0.836 & 0.958 \\ 
  $1\cdot 10^5$ & RAD & 0.2 &  &  & 0.061 & 0.777 \\ 
  $1\cdot 10^5$ & SRHT & 0.1 &  &  & 0.025 & 0.964 \\ 
  $1\cdot 10^5$ & SRHT & 0.2 &  &  & 0.171 & 0.617 \\ 
  $1\cdot 10^5$ & CW & 0.1 &  &  & 0.056 & 3.844 \\ 
  $1\cdot 10^5$ & CW & 0.2 &  &  & 2.624 & 2.937 \\ 
\bottomrule
\end{tabular}
\end{center}
\caption{Sum of squared distances between posterior mean values of the original model and models based on the respective sketches}
\label{tab:absdevfulld50}
\end{table}

In addition to the comparison to the original models' mean, we also compare the posterior means to the true means. Table \ref{tab:absdevtrued50} contains the sum of the squared distances between the true mean for $d=50$ and varying values of $\varsigma$. The general picture looks very similar to the results in Table \ref{tab:absdevfulld50}. The original model often exhibits the smallest sum of squared distances, but sometimes models based on embedded data sets are closer to the true mean. Again, there does not seem to be a systematic difference between the sketching methods. The squared distances do not seem to be influenced by the value of $n$, with some squared distances even exhibiting smaller values for larger $n$.


\begin{table}[tb]
\begin{center}
\begin{tabular}{lllrrrr}
  \toprule
n & sketch & $\varepsilon$ & $\varsigma=1$ & $\varsigma=2$ & $\varsigma=5$ & $\varsigma=10$ \\ 
  \midrule
$5 \cdot 10^4$ & none &  & 0.000 & 0.003 & 0.065 & 4.614 \\ 
  $5 \cdot 10^4$ & RAD & 0.1 & 0.048 & 0.016 & 0.124 & 1.718 \\ 
  $5 \cdot 10^4$ & RAD & 0.2 & 0.012 & 0.710 & 0.506 & 10.845 \\ 
  $5 \cdot 10^4$ & SRHT & 0.1 & 0.001 & 0.018 & 0.032 & 3.372 \\ 
  $5 \cdot 10^4$ & SRHT & 0.2 & 0.005 & 0.059 & 0.046 & 8.721 \\ 
  $5 \cdot 10^4$ & CW & 0.1 & 0.022 & 0.011 & 0.046 & 6.474 \\ 
  $5 \cdot 10^4$ & CW & 0.2 & 0.014 & 0.056 & 0.089 & 1.870 \\ 
  $1 \cdot 10^5$ & none &  &  &  & 0.065 & 0.035 \\ 
  $1 \cdot 10^5$ & RAD & 0.1 & 0.007 & 0.031 & 1.354 & 0.679 \\ 
  $1 \cdot 10^5$ & RAD & 0.2 & 0.033 & 0.009 & 0.117 & 0.579 \\ 
  $1 \cdot 10^5$ & SRHT & 0.1 & 0.030 & 0.136 & 0.040 & 0.696 \\ 
  $1 \cdot 10^5$ & SRHT & 0.2 & 0.007 & 0.125 & 0.387 & 0.453 \\ 
  $1 \cdot 10^5$ & CW & 0.1 & 0.004 & 0.232 & 0.022 & 4.496 \\ 
  $1 \cdot 10^5$ & CW & 0.2 & 0.011 & 0.072 & 3.484 & 3.473 \\ 
  $5 \cdot 10^5$ & RAD & 0.1 & 0.009 & 0.223 & 0.563 & 12.920 \\ 
  $5 \cdot 10^5$ & RAD & 0.2 & 0.045 & 0.322 & 1.729 & 0.658 \\ 
  $5 \cdot 10^5$ & SRHT & 0.1 & 0.009 & 0.147 & 0.418 & 0.059 \\ 
  $5 \cdot 10^5$ & SRHT & 0.2 & 0.016 & 0.033 & 0.085 & 2.978 \\ 
  $5 \cdot 10^5$ & CW & 0.1 & 0.027 & 0.097 & 1.305 & 0.153 \\ 
  $5 \cdot 10^5$ & CW & 0.2 & 0.050 & 0.009 & 0.135 & 3.579 \\ 
  $1 \cdot 10^6$ & RAD & 0.1 & 0.001 & 0.016 & 0.126 & 3.967 \\ 
  $1 \cdot 10^6$ & RAD & 0.2 & 0.080 & 0.011 & 0.072 & 1.357 \\ 
  $1 \cdot 10^6$ & SRHT & 0.1 & 0.002 & 0.010 & 0.599 & 0.288 \\ 
  $1 \cdot 10^6$ & SRHT & 0.2 & 0.002 & 0.183 & 2.029 & 4.329 \\ 
  $1 \cdot 10^6$ & CW & 0.1 & 0.002 & 0.289 & 1.202 & 4.445 \\ 
  $1 \cdot 10^6$ & CW & 0.2 & 0.003 & 0.047 & 0.100 & 0.395 \\ 
   \bottomrule
\end{tabular}
\end{center}
\caption{Sum of squared distances between true mean values and posterior means of models based on the respective sketches}
\label{tab:absdevtrued50}
\end{table}

\subsection{Comparison of fitted values}
\label{sec:simstudy:fitted}

After this comparison on the level of parameters -- whose number is not changed by sketching -- we will compare the models on the level of observations, of which the sketches contain merely a fraction of the number of observations in the original data set. We multiply $X$ with the posterior mean vector of $\beta$, where this posterior mean can be based on the original data set or on the respective sketches. In a frequentist sense, these are fitted values $\hat Y$, but all $X$ values are taken from the original data set, not necessarily from the data set the model is based on. This is done to see how close the approximation is on the level of $Y$ values for both $\varepsilon=0.1$ and $\varepsilon=0.2$. Figure \ref{fig:fv244CvsF} is a scatterplot which contains smoothed densities. The fitted values based on the original model are on the $x$-axis while the fitted values based on the CW sketch (with $\varepsilon=0.1$) are on the $y$-axis. Darker shades of black stand for more observations. Even though the fitted values are based on one of the data sets with the highest standard deviation of the error  ($n=50\,000, \varsigma=10$), all values are close or reasonably close to the bisecting line. This means that the fitted values obtained by the two models do not differ by much. To get a better overview, Figure \ref{fig:fvdiff244vsFbp} depicts the distances between the fitted values as boxplots. Here, all three sketching methods with both $\varepsilon=0.1$ and $\varepsilon=0.2$ are included. All six sets of distances are centered around zero. The effect of the approximation parameter $\varepsilon$ is evident from the boxplot, the variation is larger for $\varepsilon=0.2$ regardless of the sketching method. When fixing $\varepsilon$, all three sketching methods exhibit similar results, although the RAD sketch seems to introduce slightly more variation into the differences than the other two sketching methods, thus prediction for the data used in learning the model is highly accurate.

\begin{figure}[tb]
\begin{center}
\includegraphics[width = 0.375\textwidth]{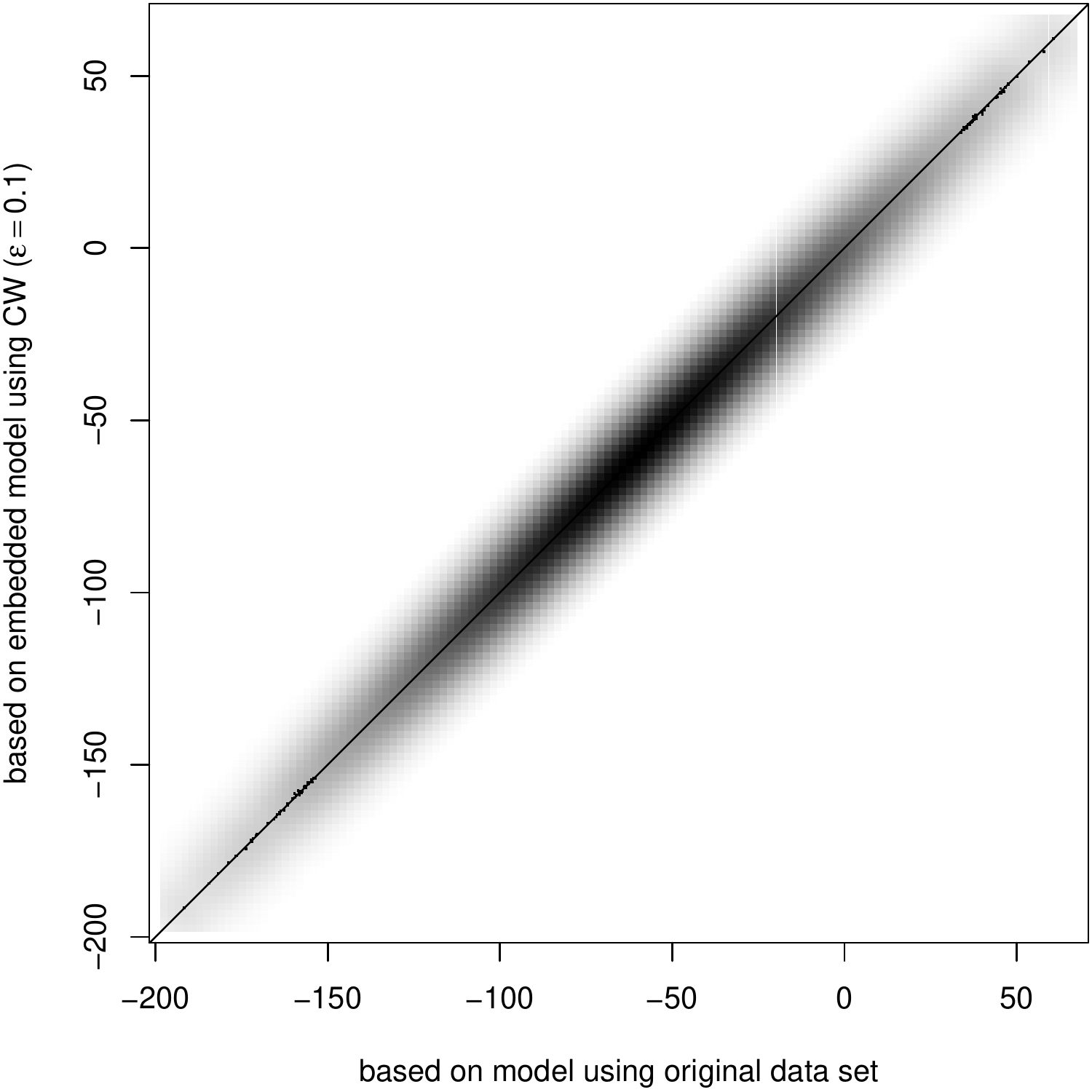}
\end{center}
\caption{Comparison of fitted values based on the original data set with $n=50\,000$, $d=50$, $\varsigma=10$ and a CW sketch with $\varepsilon=0.1$. Darker shades of black stand for more observations}
\label{fig:fv244CvsF}
\end{figure}

\begin{figure}[tb]
\begin{center}
\includegraphics[width = 0.75\textwidth]{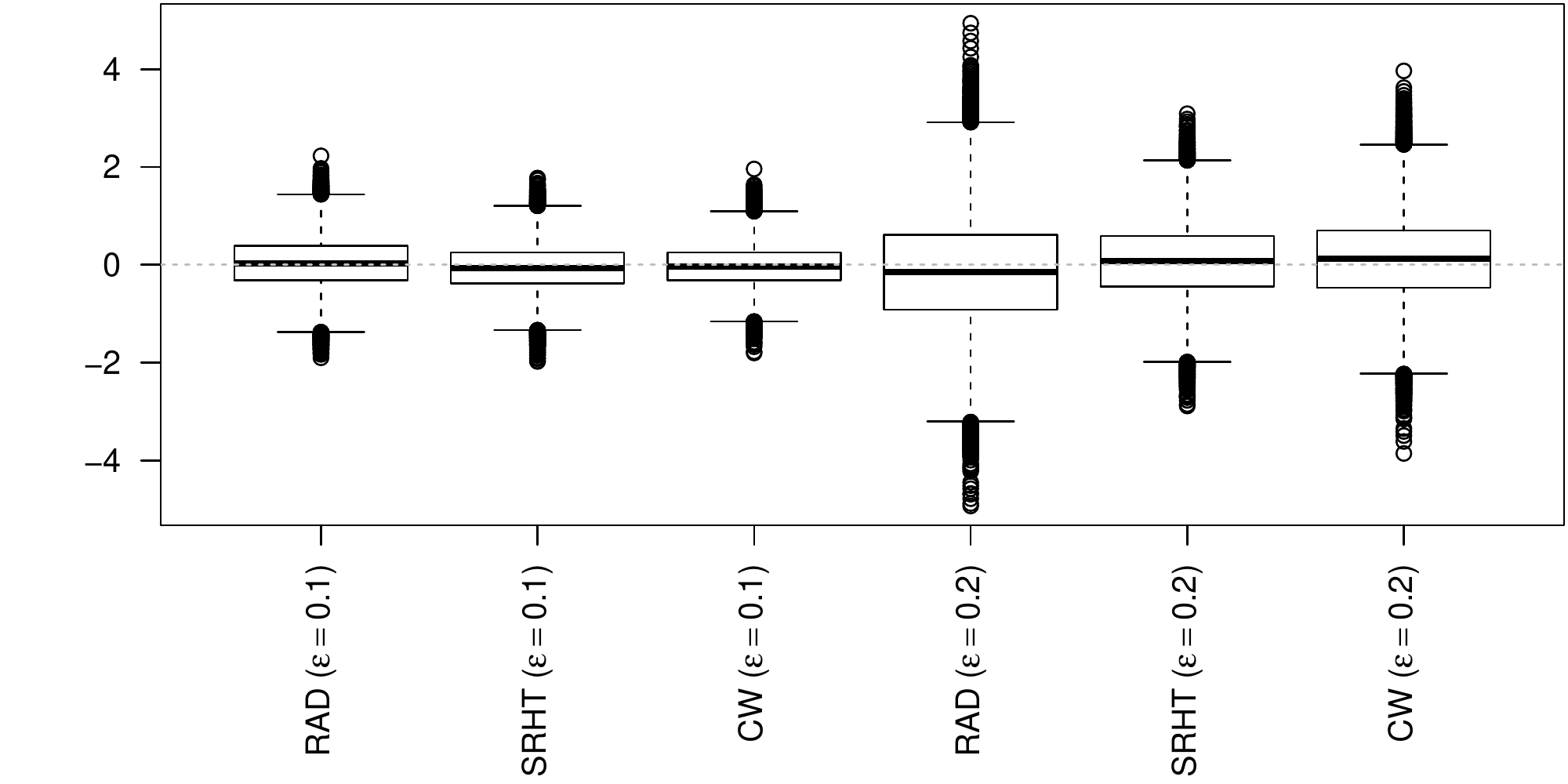}
\end{center}
\caption{Difference of fitted values according to models based on the respective sketching methods and fitted values according to model based on original data set with $n=50\,000$, $d=50$, $\varsigma=10$}
\label{fig:fvdiff244vsFbp}
\end{figure}

We have also generated additional data that has not been used in learning the model and employed the posterior mean to predict $y$-values for these data. The results are very similar to those described above and presented in Figures \ref{fig:fv244CvsF} and \ref{fig:fvdiff244vsFbp}. Sketching, again, introduces a little more variation, depending on $\varepsilon$. More formal treatment of prediction accuracy based on similar sketching techniques for the OLS solution is given in \cite{RaskuttiM15}.

\subsection{Comparison of posterior distributions}
\label{sec:simstudy:postdist}

As we have conducted Bayesian regression the model consists not only of a mean value, but of a whole posterior distribution for each parameter. Figure \ref{fig:MCMCdist243beta11_33} contains two exemplary boxplots of MCMC samples representing marginal posterior distributions. The original data set contains $n=50\,000$ observations, $d=50$ variables and has an error standard deviation of $\varsigma=5$. The medians of the MCMC samples based on the original data set are well-represented by the MCMC samples based on sketches. Even though the median based on an embedded data set might be higher or lower for certain parameters, we did not find any systematic biases. The embedding introduces additional variation, which depends on the value of the approximation parameter $\varepsilon$, but does not seem to be influenced by the choice of sketching method.

In regression, a common task is the identification of important variables by means of variable selection. In a Bayes\-ian setting, this can be done via credible intervals, among other approaches. Our results indicate that the identification of important variables is quite accurately possible based on the resulting approximate models. However, one has to take the additional variation into account. Exemplarily, when using 95\% credible intervals as criterion, one should not compare the endpoints of the credible interval to a fixed value $\mu$. Instead, take the extra variation in the posterior distribution and also possible small shifts of the mean and median into account. We therefore recommend using smaller values of $\varepsilon$ when aiming at variable selection (see Figure \ref{fig:MCMCdist243beta11_33}).
%

\begin{figure}[tb]
\begin{center}
\includegraphics[width = 0.75\textwidth]{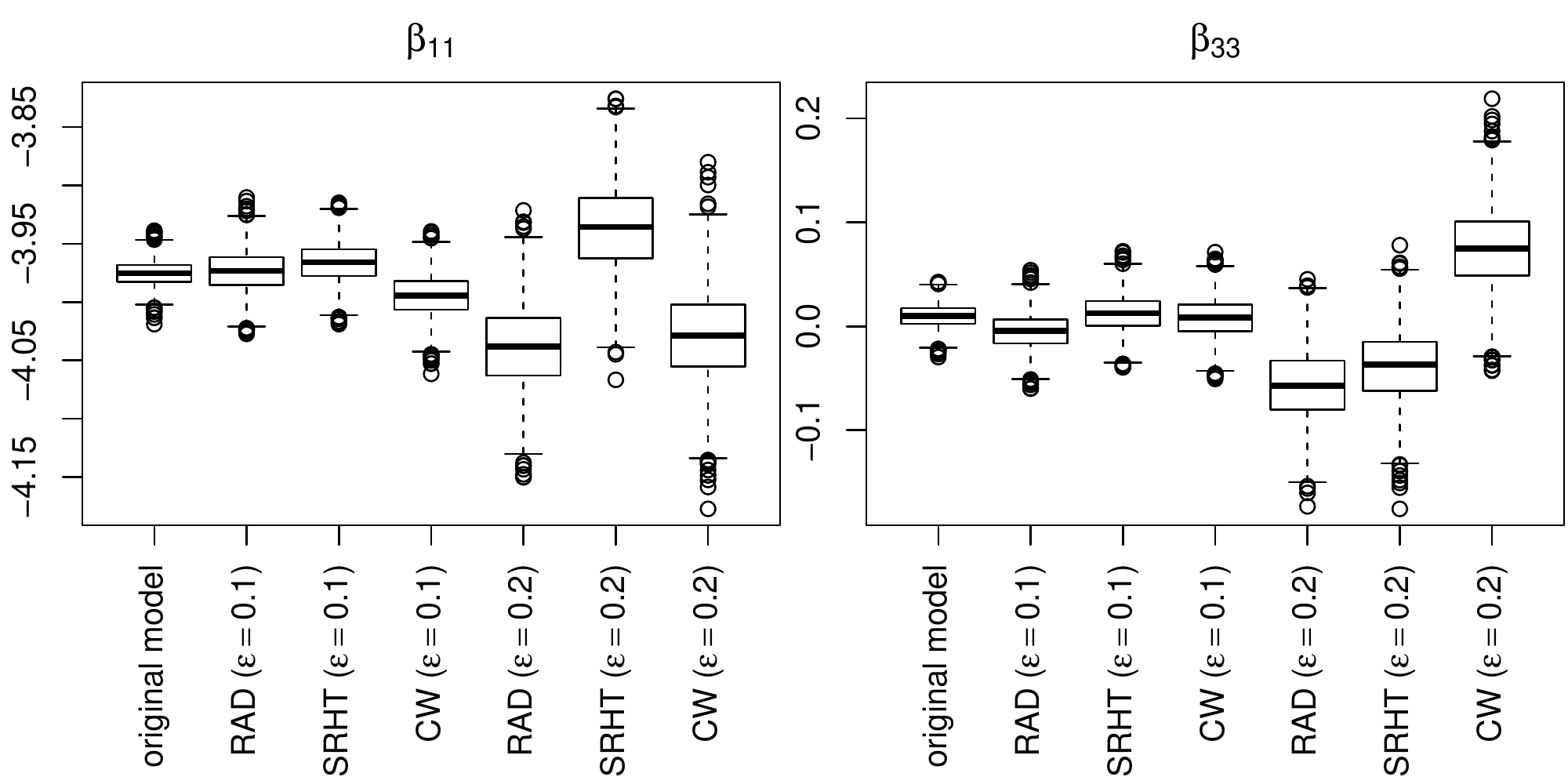}
\end{center}
\caption{Boxplots of MCMC sample for two parameters based on data set with $n=50\,000$, $d=50$, $\varsigma=5$ and respective sketches}
\label{fig:MCMCdist243beta11_33}
\end{figure}

\subsection{Comparison of running times}
\label{sec:simstudy:runningtimes}

\begin{table}[p]
\begin{center}
\begin{tabular}{rlrp{0.05cm}rrrrp{0.05cm}rrrr}
  \toprule
 & & & & \multicolumn{4}{c}{Preprocessing} & & \multicolumn{4}{c}{Analysis}\\
n & sketch & $\varepsilon$ & & $\varsigma=1$ & $\varsigma=2$ & $\varsigma=5$ & $\varsigma=10$ & & $\varsigma=1$ & $\varsigma=2$ & $\varsigma=5$ & $\varsigma=10$\\
  \midrule
 $5\cdot 10^4$ & none & & & 0.32 & 0.41 & 0.43 & 0.44 & & 1095.55 & 749.10 & 616.47 & 498.68 \\ 
  $5\cdot 10^4$ & RAD & 0.1 & & 1.60 & 1.68 & 1.68 & 1.73 & & 315.12 & 213.42 & 156.79 & 154.73 \\ 
  $5\cdot 10^4$ & RAD & 0.2 & & 0.40 & 0.39 & 0.42 & 0.43  & & 23.17 & 26.00 & 17.48 & 21.81 \\ 
  $5\cdot 10^4$ & SRHT & 0.1 & & 0.03 & 0.02 & 0.03 & 0.03 & & 317.34 & 278.39 & 181.94 & 166.41 \\ 
  $5\cdot 10^4$ & SRHT & 0.2 & & 0.02 & 0.02 & 0.02 & 0.02 & & 29.04 & 30.65 & 23.26 & 26.00 \\ 
  $5\cdot 10^4$ & CW & 0.1 & & 0.01 & 0.01 & 0.01 & 0.01 & & 375.22 & 293.89 & 164.56 & 171.82 \\ 
  $5\cdot 10^4$ & CW & 0.2 & & 0.01 & 0.01 & 0.01 & 0.01 & & 26.92 & 25.77 & 20.57 & 22.94 \\ 
  $1\cdot 10^5$ & none & & & 0.69 & 0.83 & 1.02 & 1.05 & &  &  & 2035.81 & 1617.28 \\ 
  $1\cdot 10^5$ & RAD & 0.1 & & 3.27 & 3.41 & 3.41 & 3.27 & & 278.87 & 260.80 & 167.24 & 182.92 \\ 
  $1\cdot 10^5$ & RAD & 0.2 & & 0.76 & 0.84 & 0.84 & 0.80 & & 21.44 & 23.21 & 17.52 & 23.65 \\ 
  $1\cdot 10^5$ & SRHT & 0.1 & & 0.05 & 0.06 & 0.05 & 0.05 & & 284.96 & 282.20 & 128.60 & 196.48 \\ 
  $1\cdot 10^5$ & SRHT & 0.2 & & 0.04 & 0.04 & 0.05 & 0.04 & & 23.72 & 26.82 & 21.52 & 22.70 \\ 
  $1\cdot 10^5$ & CW & 0.1 & & 0.02 & 0.03 & 0.02 & 0.02 & & 257.50 & 278.20 & 186.95 & 198.78 \\ 
  $1\cdot 10^5$ & CW & 0.2 & & 0.02 & 0.02 & 0.02 & 0.02 & & 21.94 & 26.29 & 21.22 & 23.45 \\ 
  $5\cdot 10^5$ & none &  & &  5.49 & 5.16 & 5.92 & 5.71 &&   &  &  & \\ 
  $5\cdot 10^5$ & RAD & 0.1 & & 16.88 & 15.96 & 16.10 & 16.36 & & 279.81 & 313.33 & 165.85 & 198.40 \\ 
  $5\cdot 10^5$ & RAD & 0.2 & & 3.73 & 4.00 & 4.03 & 3.85 & & 27.37 & 27.19 & 17.22 & 19.58 \\ 
  $5\cdot 10^5$ & SRHT & 0.1 & & 0.20 & 0.21 & 0.20 & 0.21 & & 310.20 & 308.01 & 190.78 & 190.18 \\ 
  $5\cdot 10^5$ & SRHT & 0.2 & & 0.19 & 0.18 & 0.19 & 0.19 & & 31.37 & 25.62 & 22.26 & 24.76 \\ 
  $5\cdot 10^5$ & CW & 0.1 & & 0.09 & 0.09 & 0.09 & 0.10 & & 335.74 & 300.32 & 189.33 & 166.92 \\ 
  $5\cdot 10^5$ & CW & 0.2 & & 0.09 & 0.08 & 0.09 & 0.09 & & 26.03 & 25.23 & 24.39 & 22.86 \\ 
  $1\cdot 10^6$ & none &  & &  18.23 & 12.88 & 12.59 & 14.09 & &  &  &  &  \\ 
  $1\cdot 10^6$ & RAD & 0.1 & & 51.77 & 147.42 & 33.75 & 34.71 & & 209.19 & 279.03 & 215.78 & 145.64 \\ 
  $1\cdot 10^6$ & RAD & 0.2 & & 7.92 & 8.46 & 8.38 & 8.21 & & 21.27 & 19.93 & 22.87 & 23.43 \\ 
  $1\cdot 10^6$ & SRHT & 0.1 & & 0.41 & 0.49 & 0.61 & 0.62 & & 341.12 & 264.99 & 294.04 & 154.77 \\ 
  $1\cdot 10^6$ & SRHT & 0.2 & & 0.39 & 1.44 & 0.68 & 0.68 & & 26.61 & 31.32 & 19.69 & 23.69 \\ 
  $1\cdot 10^6$ & CW & 0.1 & & 0.19 & 0.27 & 0.38 & 0.46 & & 281.72 & 232.40 & 175.49 & 144.02 \\ 
  $1\cdot 10^6$ & CW & 0.2 & & 0.21 & 0.19 & 0.45 & 0.39 & & 28.58 & 19.50 & 22.05 & 9.72 \\ 
   \bottomrule
 \end{tabular}
\caption{Running times for data sets with $d = 50$. Columns 4 to 7 (``Preprocessing'') contain the running times of the sketching methods in minutes, for the original data set, the values represent the time required to read the data set into memory, which is a prerequisite for every sketching method. The four columns to the right (``Analysis'') contain the running times of the Bayesian linear regression analysis in minutes}
\label{tab:completetimed50}
\end{center}
\end{table}

One of our aims is to make Bayesian regression feasible on very large data sets. After ensuring that the results are close to those obtained by the analysis on the original data set, we will now contemplate the running time required. The total running time is composed of the time required for the analysis and the time required for the necessary preliminary steps: reading the data set into memory and calculating the sketch. Table \ref{tab:completetimed50} contains the running times required for the Bayesian regression analysis (four rightmost columns headed ``Analysis'') and the running times for the preliminary steps (columns headed ``Preprocessing''). For the original data sets, the preliminary steps consist of the time required for reading the data set into memory, for all other cases, Table \ref{tab:completetimed50} gives the running times required for constructing the sketch. The total running time of an analysis on a sketch is obtained by summing up the reading time, the sketching time, and the time required for the Bayesian linear regression analysis.

Table \ref{tab:completetimed50} suggests that both the reading times and the sketching times are stable for data sets of the same size, with the possible exception of an outlier for $n=1\,000\,000$ and $\varsigma=2$. Both the reading times and the sketching times grow with the size of the data set. The sketching times grow for smaller values of $\varepsilon$. For RAD sketches with $\varepsilon=0.1$, the sketching takes longer than the reading of the data set, for all other combinations the opposite is true. CW sketches require the shortest amount of running time of the three sketching methods. 

Although only a few of the original data sets could be analyzed, the ``Analysis'' values in Table \ref{tab:completetimed50} indicate that the running times for the Bayesian analysis increase with the number of observations in the data set. The running times for the sketched data sets on the other hand show no systematic increase for growing values of $n$.  There is some variation in the running times, but this seems to be more random chance than trend. Larger values of $\varepsilon$ lead to shorter running times in the analysis, which indicates that the trade-off between time and goodness of the approximation is present both in calculating and analyzing the sketch. The running time of the analyses is similar for all three sketching methods. However, the running times do seem to depend on the value of $\varsigma$. For higher standard deviations of the error term, the required running time tends to become less.


\begin{figure}[bt]
\begin{center}
\includegraphics[width = 0.75\textwidth]{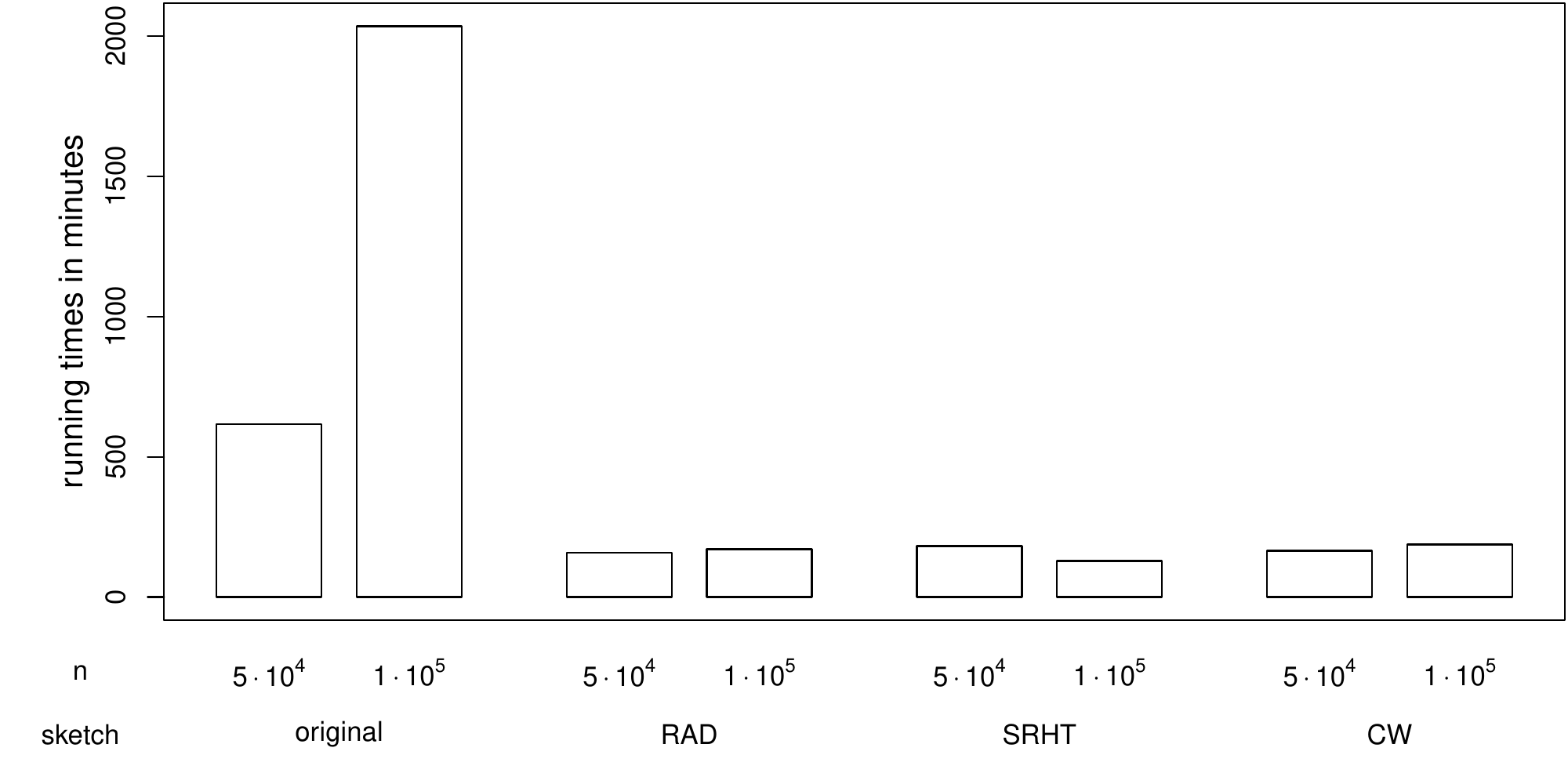}
\end{center}
\caption{Total running times in minutes for data sets with $n\in\{50\,000, 100\,000\}$, $d=50$, $\varsigma=5$ and approximation parameter $\varepsilon=0.1$. For the sketched data sets, the total running time consists of the time for reading, sketching and analyzing the data set. For the original data set, the sketching time is 0 since this step is not applied}
\label{fig:d50runningtimesvarsigma5}
\end{figure}

Figure \ref{fig:d50runningtimesvarsigma5} exemplarily shows the total running times in minutes for data sets with $d=50$ and $\varsigma=5$. This comprises reading, sketching (if applicable) and analyzing the data set. For the sketches, $\varepsilon=0.1$ was used. To illustrate the potential improvement with increasing $n$, Figure \ref{fig:d50runningtimesvarsigma5} contains the total running times for $n=50\,000$ and $n=100\,000$ side by side. For the original data sets, we observe a large increase in running time as $n$ doubles, whereas the running times on the sketched data sets hardly change.

All results presented so far are based on data sets with $d=50$ variables. We have conducted Bayesian analyses on data sets with $d=100$ with the same parameter values for $n$ and $\varsigma$. The findings from these simulations are similar to those for $d=50$. One exception is that the analysis of CW sketches now takes longer than the analysis of the respective original data sets for $n=50\,000$. This changes for larger data sets. One strength of the CW method, however, is its efficiency when dealing with streaming data when the number of variables is not too large. Recall that its dependency on $d$ is quadratic, which means that the sketch sizes are already very large for medium dimensional problems with $d\approx 500$ or $d\approx 1\,000$. Setting the target dimension to a lower value is not a remedy, because this results in a weaker approximation guarantee as also noticed by \cite{YangMM15}. In that case one should consider using one of the slower and denser sketches.

\subsection{Streaming example and concluding remarks}
\label{sec:simstudy:stream}

The data sets considered so far were chosen to be small enough to allow for Bayesian analysis on the full data set in a convenient time. This is by far not what might be called Big Data. To show that our approach is suitable for analyzing really big amounts of data, we have generated and at the same time sketched a data set. The generation of the data followed the same rules as the other simulated data sets, but with $\varsigma = 0.1$. 
The data set's original size is \mbox{$10^9\times 100$} double precision values. This corresponds to about 2 TB in CSV format or at least 750 GB in a binary representation. We have used the CW sketching method resulting in a sketch of size $65\,536\times 100$, which requires only around 140 MB of space in CSV format and fits into RAM. Clearly, we cannot compare the results to those on the original data set, but calculating the sum of squared distances between the true mean values of $\beta$ and the posterior mean of the model adds up to $3.741\cdot 10^{-6}$. The Bayesian regression analysis took $2781$ minutes.

In some cases the algebraic structure might strongly depend on a few observations or variables. In such situations it is important to identify these or to retain their contribution in the reduced data set. So far, our model assumptions did not suffer from such ill-behaved situations, but now we assess the performance of our method in this case. We construct data sets where an important part of the target variable falls into a subspace that is spanned by a small constant number of observations. Uniform random subsampling will pick these only with probability $O(\frac{1}{n})$. Oblivious subsampling techniques in general will have trouble identifying the important observations. In contrast, oblivious subspace embedding techniques preserve these effects with high probability. This effect is observed in practice even when comparing one sketch against the best of $1000$ repetitions of uniform random subsampling.

In conclusion, the simulation study indicates that our proposed method works well for simulated data sets, which are generally well-suited for conducting Bayesian linear regression. But even with a high variance of the error term (and thus a relatively bad model fit), our proposal leads to results similar to those one would obtain on the original data set. The running time of the analysis with the proposed sketches is largely independent of $n$, giving advantages for very large $n$. Since the embeddings can be read in sequentially, it is not necessary to load the whole data set into the memory at once, which reduces the required memory.

For CW embeddings, reading the data in and calculating the sketch only takes marginally longer than only reading the data in. In our experiments, we found that reading in and sketching takes around $1.01$ to $1.04$ times longer. This factor is typically higher for small data sets and lower for larger data sets (cf. Table \ref{tab:completetimed50}). However, when the number of variables is large it may be favorable to use SRHT, whose sketching time is only slightly larger but has considerably smaller embedding dimension.



\section{Real Data Example}
\label{sec:realdata}

As a real data example, we consider the bike sharing data set taken from \cite{Gama14}, which is available in the UCI Machine Learning Repository \citep{Lichman:2013}. This is only meant as an exemplary application of the methods to a real data scenario and should not be mistaken for a complete statistical analysis of the data set. The bike sharing data set contains the number of rental bike users per hour over two years as well as additional information about the day and the weather. See Table \ref{tab:bikevars} for an overview  of the variables we use in the model. The data set contains some additional variables we do not employ, because they are highly correlated with the variables present in the model. We also made a change to the factor levels of the variable \emph{weathersit}. In the original data set, this variable contains 4 levels. The fourth level only is present 3 times out of total of $n=17\,379$ hours in the data set. To avoid any problems with such an underrepresented level, we combine levels 3 and 4 to obtain a factor with 3 levels. The original levels 3 and 4 stand for ``light rain'' and ``heavy rain''. The new level 3 can easily be interpreted as the presence of rain. For a more detailed description of all variables, please refer to the data set's web page on the UCI Machine Learning Repository. \footnote{http://archive.ics.uci.edu/ml/datasets/Bike+Sharing+Dataset}

\begin{table}[tb]
\begin{center}
\begin{tabular}{llll}
\toprule
Variable & Description & Remark\\
\midrule
\emph{cnt} & number of rental bikes used & $Y$-variable\\
\emph{season} & season of the year & factor (4 levels)\\
\emph{yr} & year (2011 or 2012) & factor (2 levels)\\
\emph{hour} & hour (0 to 23) & factor (24 levels)\\
\emph{holiday} & public holiday & factor (2 levels)\\
\emph{weekday} & day of the week & factor (7 levels)\\
\emph{weathersit} & weather (``clear'' to ``rain'') & factor (3 levels)\\
\emph{atemp} & apparent temperature & standardized\\
\emph{hum} & humidity & standardized\\
\emph{windspeed} & windspeed & standardized\\
\bottomrule
\end{tabular}
\end{center}
\caption{Variables from the bike sharing data set used in the model}
\label{tab:bikevars}
\end{table}%

The variable \emph{cnt} contains the number of rental bikes used per hour and is thus a count variable. However, there are around 850 distinct values, which makes analyzing \emph{cnt} as a continuous variable reasonable. When analyzing such variables with a linear regression model, transforming them using the square-root is a common procedure. After the transformation, the values of \emph{cnt} show some bi-modality, but fit reasonably well to the assumption of a normal distribution.

We use the transformed variable \emph{cnt} as $Y$-variable and all other variables in Table \ref{tab:bikevars} as $X$-variables. To handle the factor variables, we use the R-function \texttt{model.matrix} to create a design matrix, which is then passed on to \texttt{RaProR} and \texttt{rstan}. The resulting design matrix contains $n=17\,379$ observations and $d=39$ variables plus the intercept. Again, we calculate sketches for all three methods and with two different settings of $\varepsilon$. Because of the size of the data set relative to the number of variables, we choose $\varepsilon=0.15$ and $\varepsilon=0.2$ for the RAD and SRHT sketches. For CW, we choose values of $k$ that are closest to the target dimension of the other sketches. This results in the values given in Table \ref{tab:realsketchsize}.

\begin{table}[tb]
\begin{center}
\begin{tabular}{rrrrr}
\toprule
$d$ & $\varepsilon$ & RAD & SRHT & CW\\
\midrule
40 & 0.15 & $6\,767$ & $6\,767$ & $8\,192$\\
40 & 0.20 & $3\,807$ & $3\,807$ & $4\,096$\\
\bottomrule
\end{tabular}
\end{center}
\caption{Number of observations of the sketches for the bike sharing example. Different values of $\varepsilon$ are used for RAD and SRHT sketches; the target dimension of CW sketches is chosen to be the power of two closest to the size of the RAD and SRHT sketches} 
\label{tab:realsketchsize}
\end{table}%

The Bayesian model based on the original data set suggests that all mentioned variables are important for the modeling of the number of bikes used per hour. Figure \ref{fig:app:bikebppostall} in the appendix gives an overview of the posterior distributions for the model based on the original data set. As one might expect, the weather has a strong influence. More bikes are rented when the apparent temperature is high and, to a lesser extent, when the humidity and the wind speed are low. In clear or partly cloudy weather, the number of rented bikes is highest, but the negative influence of heavier clouds is comparatively small. Rainy weather, however, reduces the number of bike users more substantially. In addition to that, fall seems to be the most popular seasons for bike sharing. Spring and summer also have positive effects in comparison to winter, but the effect sizes are smaller. This might seem surprising at first, especially as the number of rental bike users is highest in summer. This might partly be an effect of the apparent temperature, which is generally higher in summer.

There is also a distinct hourly effect. During night time, especially between midnight and 5am, the number of rented bikes is greatly reduced. On the other hand, between 7am and 9pm, a lot of bikes are used, with two peaks at 8am and 5pm/6pm. This might indicate that the service is used by people transiting to and from work. Holidays -- which only includes days that would otherwise be a working day, so only Monday to Friday -- have a negative influence on the number of rental bike users. When taking the days of the week into account, Friday and Saturday have the highest positive effect while Sunday seems to be the least popular day. All of these effects based on days have a small influence compared to the variables mentioned before. Lastly, the variable \emph{yr} also has a positive influence which indicates a positive trend for this bike sharing service.


All conclusions can similarly be derived from the models based on the sketches. Following our approach in Section \ref{sec:simstudy}, we first compare the resulting posterior mean values of $\beta$. Table \ref{tab:absdevfullbike} shows the sum of squared distances between posterior mean values of the original model and models based on the three sketching methods, using two values of $\varepsilon$ each. There is a general increase in the sum of squared distances for $\varepsilon=0.2$ compared to $\varepsilon=0.15$, but the amount differs depending on the sketching method. This should not be over-interpreted, however. As these values represent only one realization of a random subspace embedding, there is no evidence for systematic differences.

\begin{table}[tb]
\begin{center}
\begin{tabular}{lrrr}
  \toprule
$\varepsilon$ & RAD & SRHT & CW\\
  \midrule
0.15 & 1.790 & 2.349 & 0.907 \\ 
  0.2 & 6.511 & 2.732 & 1.657 \\ 
  \bottomrule
\end{tabular}
\end{center}
\caption{Sum of squared distances between posterior mean values of the original model and models based on the respective sketches for the bike sharing data set}
\label{tab:absdevfullbike}
\end{table}

To see the effects of the differences in the posterior means of $\beta$ on the level of the $y$-variable, which is the number of rental bikes used, we compare the fitted values as in Section \ref{sec:simstudy:fitted}. Again, we multiply the original design matrix $X$ with the posterior means of $\beta$, where the posterior mean values are obtained from the model on the original design matrix and the models on the respective sketches. Figure \ref{fig:fvdiffbikevsFbp} contains the six resulting boxplots. As in the simulation study, the differences of the fitted values are centered around zero with only small deviations. Further analysis indicates that the higher deviations occur when the number of bikes used is high, which means that the majority of differences in the fitted values are small relative to the observed value for that data point.

\begin{figure}[tb]
\begin{center}
\includegraphics[width = 0.75\textwidth]{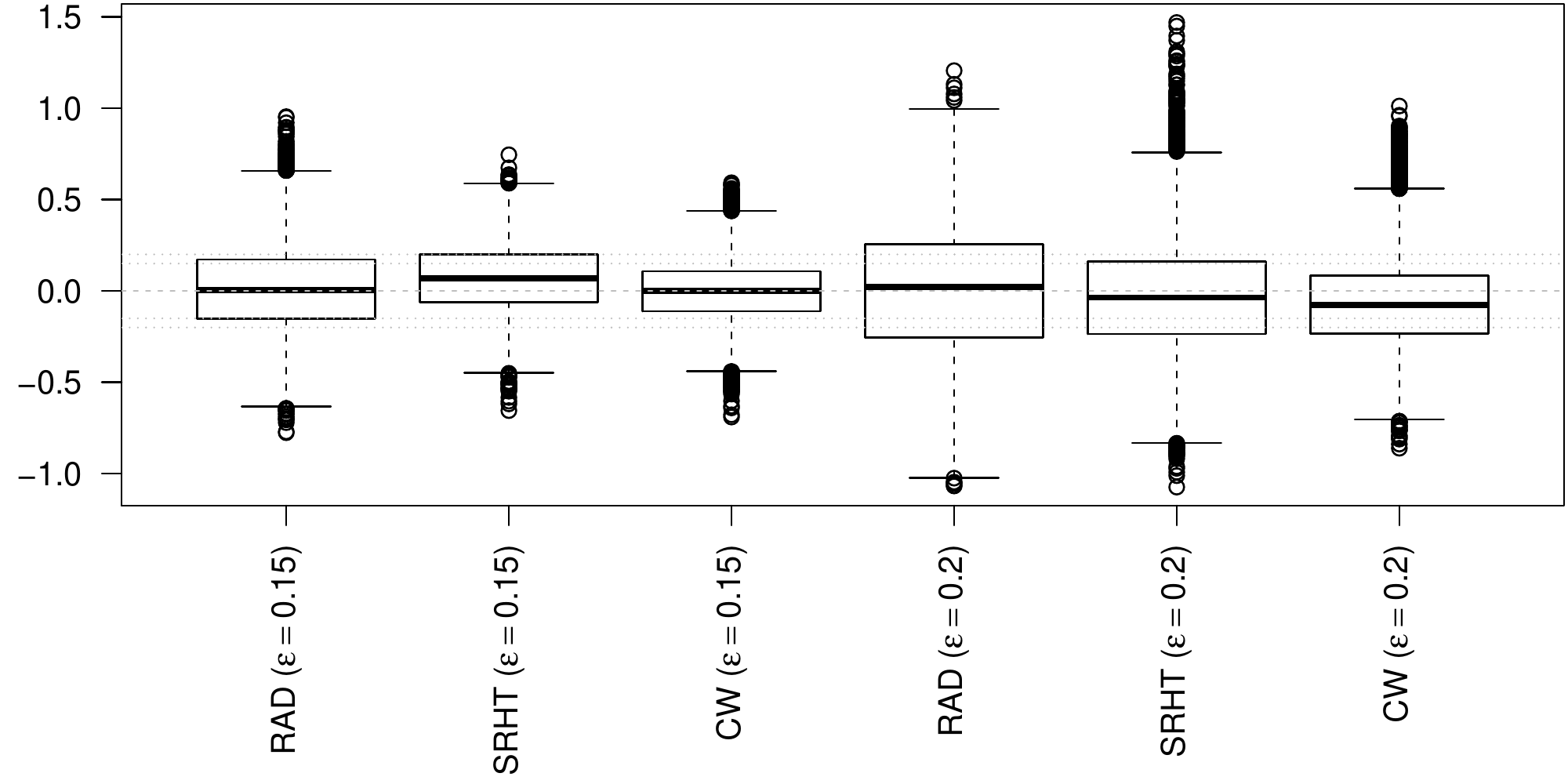}
\end{center}
\caption{Difference of fitted values according to models based on the respective sketching methods and fitted values according to model based on the original bike sharing data set}
\label{fig:fvdiffbikevsFbp}
\end{figure}

As a last step of our short analysis of the bike sharing data set, we will concentrate on the posterior distributions of the factors that take the weather situation into account. This factor has three levels in our data set. The first level stands for clear weather, which also includes partly cloudy hours, the second level stands for heavier clouds or mist while the third level models rainy weather, which includes light rain, heavy rain, thunderstorms, and snow. The different levels occur with different relative frequencies: around 66\% of the observed hours fall into level 1, 26\% into level 2 and 8\% into level 3.

\begin{figure}[tb]
\begin{center}
\includegraphics[width = 0.75\textwidth]{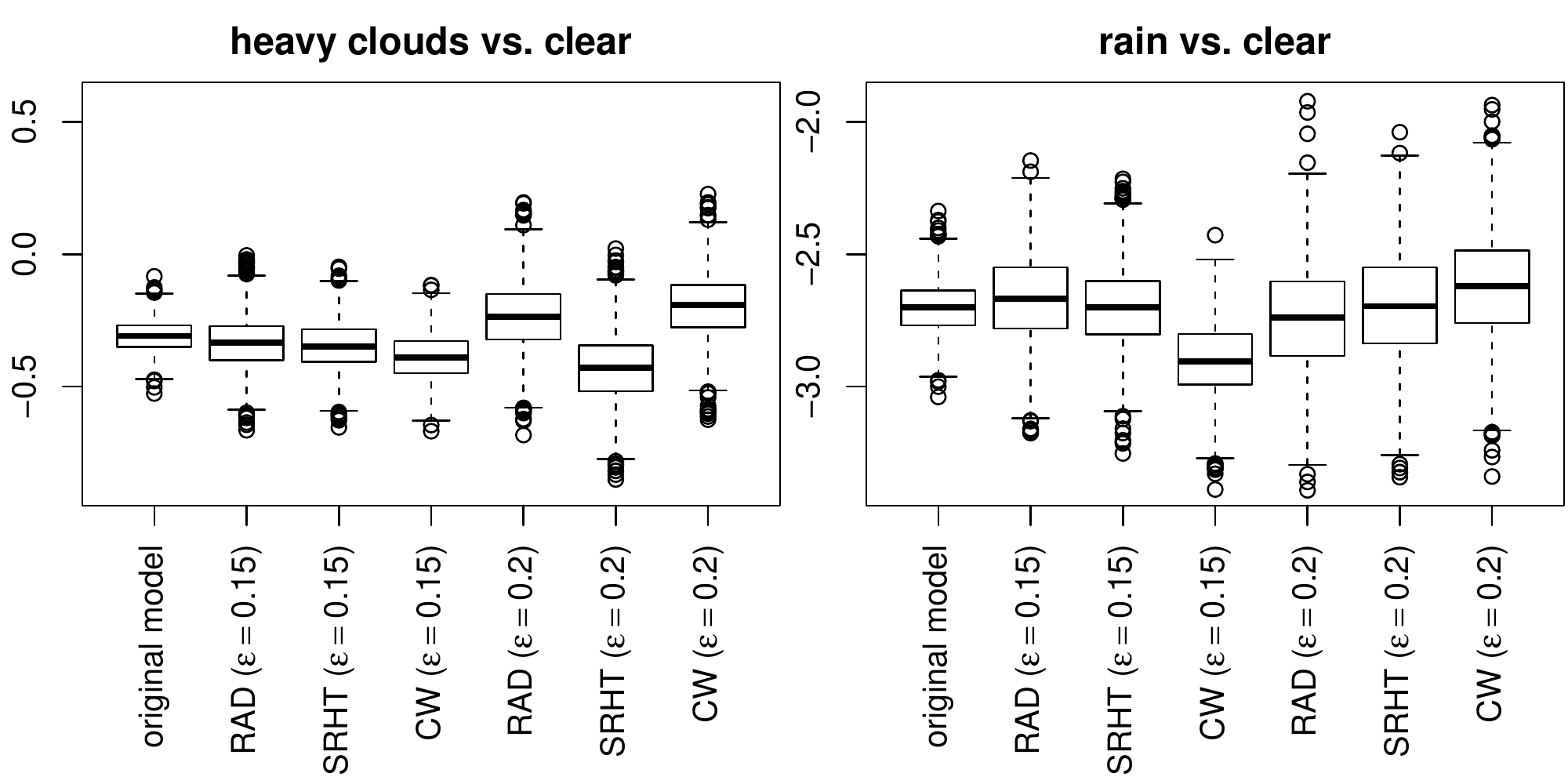}
\end{center}
\caption{Boxplots of MCMC sample for the two weather situation parameters based on the original data set and all sketches}
\label{fig:MCMCdistbikeweather}
\end{figure}

Figure \ref{fig:MCMCdistbikeweather} shows boxplots of the MCMC samples based on the original design matrix and the sketches. The values represent the marginal posterior distributions of the two dum\-my variables associated with the variable \emph{weathersit}. The scales of the boxplots are chosen such that one unit is of the same length on both $y$-axes. This allows for easy comparison of the variation in the two posterior distributions. Again, we can see that the embedding introduces additional variation, which depends on the value of the approximation parameter $\varepsilon$, but does not seem to be influenced by the choice of the sketching method. In addition, the posterior distributions for the factor ``heavy clouds'' show less variation compared to the posterior distributions for ``rain''. This is as one might expect as the number of occurrences is more than three times higher for ``heavy clouds'' and a larger number of observations tends to reduce the uncertainty. Nonetheless, it is interesting to observe that the variation introduced by the embedding seems to be a factor of the variation present in the original model.

While the effect of the factor ``rain'' is undoubtedly negative according to the original model and all sketches used here, the effect of factor ``heavy clouds'' is close to zero. In the original model, ``heavy clouds'' would also be seen as an influential factor when using the 95\% credible interval as a criterion. The conclusion is the same for all sketching methods when using $\varepsilon=0.15$. However, when using $\varepsilon=0.2$ and CW, ``heavy clouds'' would be seen as not influential. This stresses again that the endpoints of credible intervals based on sketches exhibit some additional variation and inference based on them may change, depending on the variation in the original model and the choice of $\varepsilon$. If variable selection is a focus of the regression analysis, we recommend choosing reasonably small $\varepsilon$ (cf. Figure \ref{fig:MCMCdist243beta11_33} and its discussion in Section \ref{sec:simstudy:postdist}).

This example underlines that our method also works well on real world applications when the original data follows the model assumptions reasonably well.

\section{Conclusion}
\label{sec:conclusion}
Our paper deals with random projections as a data reduction technique for Bayesian regression. We have shown how projections can be applied to compress the column\-space of a given data matrix with only little distortion. The size of the reduced data set is independent of the number $n$ of observations in the original data set. Therefore, subsequent computations can operate within time and space bounds that are also independent of $n$, regardless of which algorithm is actually used. While our focus was on MCMC and the No-U-Turn-Sampler in particular, we tried INLA as well and observed a considerable reduction in running time while achieving very similar results. However, our proposed reduction method is not limited to these approaches, making it highly flexible.

The presented embedding techniques allow for fast application to the data set and do not need the embedding matrices to be stored explicitly. Thus, only very little memory is needed while sketching the data. Furthermore, we have surveyed their useful properties when the computations are performed in sequential streaming as well as in distributed environments. These scenarios are highly desirable when dealing with Big Data \citep[cf.][]{welling:2014}.

We consider the situation where the likelihood is modeled using standard linear regression with a Gaussian error term. We show that the likelihood is approximated within small error. Furthermore, if an arbitrary Gaussian distribution is used as prior distribution, the desired posterior distribution is also well approximated within small error. This includes the case of a uniform prior distribution over $\mathds R^d$, an improper, non-informative choice that is widely used \citep[cf.][]{Gelman2014}. We also show the results to be $(1+O(\varepsilon))$ approximations of the distributions of interest in the context of Bayesian linear regression. As the structure of both mean and variance is preserved up to small errors, approximate sufficient statistics for the posterior distributions are well-recovered. This gives the user all the information that is needed for Bayesian regression analysis.

In our simulation experiments, we found that the approximation works well for simulated data sets. All three sketching methods we considered lead to results that are very similar to Bayesian regression on the full data set and the true underlying values. The running time for the MCMC analysis based on the sketches is independent of the number of observations $n$. The calculation of the embedding does depend on $n$, but requires little more time than the necessity of only reading the data. This is especially true when using the CW method. But for larger dimensions a CW embedding can be too large. In such a case, the denser SRHT construction also performs very well and is preferable because of its lower dependency on $d$. RAD has even lower dependency on $d$ but takes considerably more time to calculate.

We have applied the methods to a bike sharing data set by \cite{Gama14}. The approximation also works well on this real data example, giving very similar results to the original Bayesian regression, while adding little additional variation to the posterior distributions.

For future research, we would like to generalize our results to other classes of distributions for the likelihood and to more general priors. As a first step, we have used hierarchical models involving normal, Gamma and exponential distributions as hyperpriors. For normal and Gamma distributions, the results seem promising, whereas using exponential distributions seems more challenging. The recent results on frequentist $\ell_p$ regression of \cite{WoodruffZ13} might give rise to efficient streaming algorithms also in the Bayesian regression setting. Another interesting direction would be to consider Gaussian mixtures, since they allow to approximate any continuous distribution.

In real-world applications one might exhibit the domain-specific structure to further reduce the time and space bounds when these indicate that the data itself is of low rank or allows for sparse solution vectors.

\paragraph*{Acknowledgements}
We would like two thank the anonymous referees for their valuable comments and suggestions which helped to improve this manuscript. This work has been supported by Deutsche For\-schungs\-ge\-mein\-schaft (DFG) within the Collaborative Research Center SFB 876 ``Providing Information by Resource-Constrained Analysis", project C4.

\bibliographystyle{apalike}
\bibliography{literatur}

\newpage
\appendix
\renewcommand{\thefigure}{A.\arabic{figure}}    
\setcounter{figure}{0}  
\section*{Appendix}
\label{sec:app}
\begin{figure}[h]
\begin{center}
\includegraphics[angle=270,width = 0.65\textwidth]{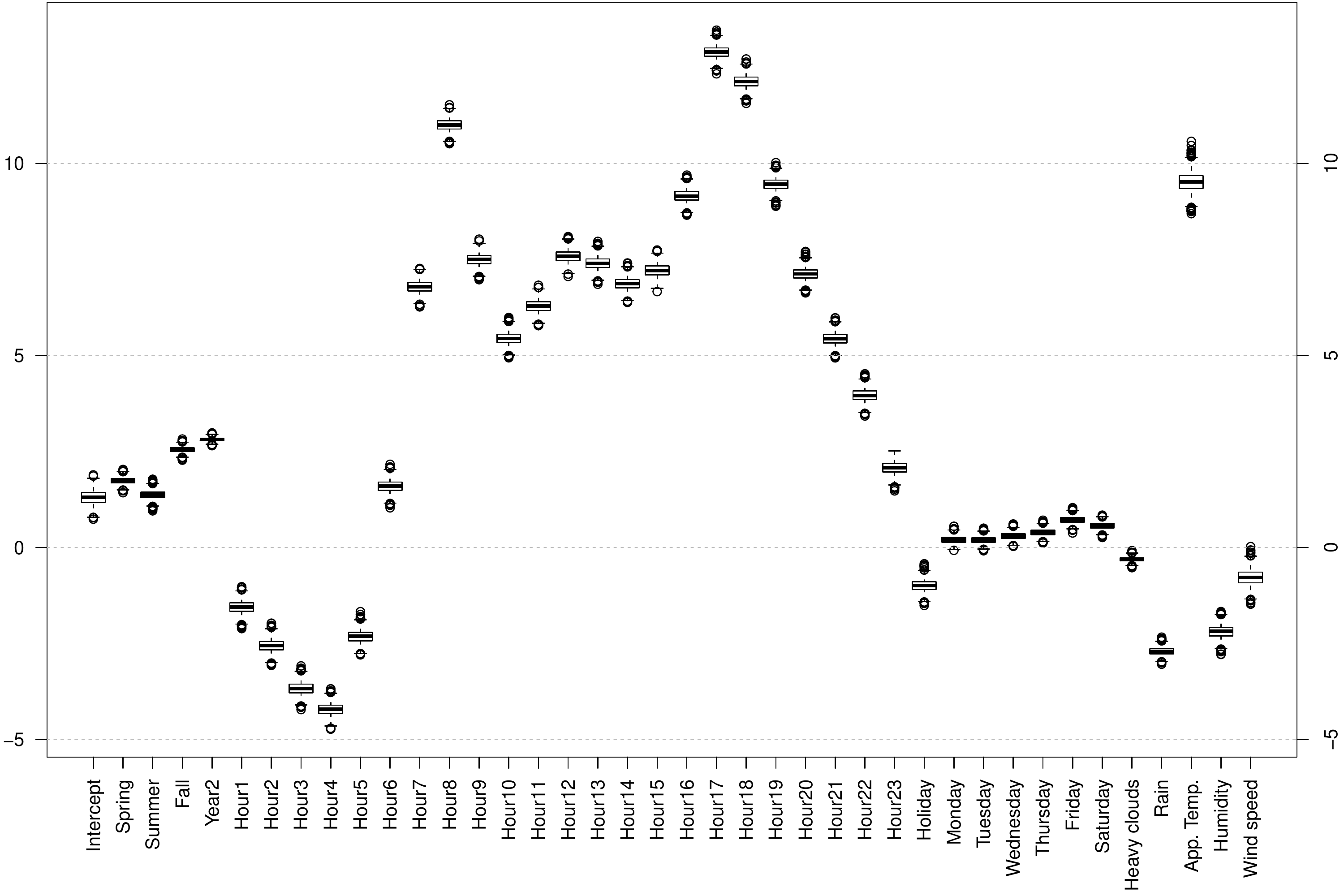}
\end{center}
\caption{Boxplots of the MCMC samples for all $\beta$ parameters of the bike sharing data set based on the original model}
\label{fig:app:bikebppostall}
\end{figure}

\end{document}